\documentclass[11pt]{article} 

\usepackage[margin=1in]{geometry} 
\usepackage[utf8]{inputenc}
\usepackage[T1]{fontenc}
\usepackage[sc]{mathpazo} 

\usepackage{xspace}

\usepackage{booktabs,multicol,multirow}
\usepackage{subcaption}

\usepackage{enumitem} 
\usepackage{float}
\usepackage[dvipsnames]{xcolor} 
\definecolor{ptblue}{RGB}{15,76,129} 
\definecolor{ptemerald}{HTML}{009473} 
\definecolor{ptgray}{HTML}{939597} 

\usepackage{tikz} 
\usetikzlibrary{positioning}

\usepackage{amsmath,amsfonts,amssymb,amsthm,bbm,mathtools}

\let\OLDland\land
\renewcommand{\land}{\:\OLDland\:}
\let\OLDlor\lor
\renewcommand{\lor}{\:\OLDlor\:}
\let\OLDforall\forall
\renewcommand{\forall}{\;\OLDforall\:}
\let\OLDexists\exists
\renewcommand{\exists}{\;\OLDexists\,}

\DeclareMathOperator*{\argmax}{arg\,max}

\usepackage[ruled,linesnumbered,vlined]{algorithm2e}

\usepackage[square,sort]{natbib} 

\usepackage{hyperref}
\hypersetup{
linktocpage,
colorlinks=true,
citecolor=ptemerald, 
urlcolor=ptblue, 
linkcolor=Plum, 
pdftitle={Optimal Subsidy Bounds for Goods and Chores: One Dollar Each Suffices},
pdfauthor={Xinhang Lu, Simon Mackenzie, Mashbat Suzuki},
}
\usepackage{aliascnt}
\usepackage{cleveref} 

\theoremstyle{plain}
\newtheorem{theorem}{Theorem}[section]
\newaliascnt{claim}{theorem}

\aliascntresetthe{claim}
\newaliascnt{conjecture}{theorem}

\aliascntresetthe{conjecture}
\newaliascnt{corollary}{theorem}

\aliascntresetthe{corollary}
\newaliascnt{lemma}{theorem}
\newtheorem{lemma}[lemma]{Lemma}
\aliascntresetthe{lemma}
\newaliascnt{proposition}{theorem}
\newtheorem{proposition}[proposition]{Proposition}
\aliascntresetthe{proposition}
\newaliascnt{observation}{theorem}
\newtheorem{observation}[observation]{Observation}
\aliascntresetthe{observation}

\theoremstyle{definition}
\newaliascnt{definition}{theorem}
\newtheorem{definition}[definition]{Definition}
\aliascntresetthe{definition}

\theoremstyle{remark}

\makeatletter 
\newcommand{\EF}[1]{\if\relax\detokenize\expandafter{\@firstofone#1{}}\relax \text{EF}\xspace\else \text{EF#1}\fi}
\makeatother
\newcommand{\EFOne}{\EF{1}\xspace}
\newcommand{\EFX}{\EF{X}\xspace}

\newcommand{\Inj}{\mathrm{Inj}}
\newcommand{\val}{\mathrm{val}}

\newcommand{\EFable}{\text{envy-freeable}\xspace}
\newcommand{\EFability}{envy-freeability\xspace}

\newcommand{\alloc}{\mathcal{A}\xspace}

\newcommand{\choreMaximality}{chore-maximality\xspace}

\newcommand{\IMWM}{\texttt{\textup{IMWM}}\xspace} 
\newcommand{\IMWPM}{\texttt{\textup{IMWPM}}\xspace} 

\title{Optimal Subsidy Bounds for Goods and Chores: \\ One Dollar Each Suffices}

\author{Xinhang Lu\thanks{\nolinkurl{xinhang.lu@inf.kyushu-u.ac.jp}} \\ Kyushu University
\and
Simon Mackenzie\thanks{\nolinkurl{simon.william.mackenzie@gmail.com}} \\ UNSW Sydney
\and
Mashbat Suzuki\thanks{\nolinkurl{mashbat.suzuki@unsw.edu.au}} \\ UNSW Sydney}

\date{}

\begin{document}
\maketitle

\begin{abstract}
We study the fair allocation of $m$ indivisible items to $n$ agents with additive utilities.
In our setting, each indivisible item may be a good, yielding non-negative utility to some agents, or a chore, yielding negative utility to others.
Whilst envy-free allocations may not exist in the indivisible-items setting, envy-freeness can be achieved if some amount of divisible good (i.e., \emph{money}) is introduced. 
When each item’s utility or disutility is bounded by one, we show that a subsidy of at most one dollar per agent suffices to guarantee the existence of an envy-free allocation, and that this bound is tight. Moreover, such an allocation can be computed in polynomial time.
Since at least one agent need not receive any subsidy, our results imply that a total subsidy of at most $n-1$ dollars suffices to ensure envy-freeness.

\end{abstract}

\section{Introduction}

Dividing scarce resources among multiple parties is ubiquitous, arising in contexts such as divorce settlements and inheritance divisions.
A natural objective is to select an allocation that every participant perceives as fair.
Among the many fairness notions studied in the literature, the gold standard is arguably \emph{envy-freeness}, introduced by \citet{Foley67}: an allocation is envy-free if no agent prefers another agent's bundle to their own.

Our focus is on allocating \emph{indivisible items}.
This includes familiar \emph{goods} (e.g., houses, artworks, electronics) as well as \emph{chores} or undesirable tasks (e.g., household chores, shift scheduling).
With indivisible items, envy-free allocations need not exist; for example, with two agents and a single valuable indivisible good, whichever agent does not receive the good necessarily envies the other.

A classical way to circumvent such non-existence is to allow \emph{monetary compensation}.
We follow the standard \emph{subsidy} model, where an outcome consists of an allocation of the indivisible items together with non-negative payments to agents, and envy-freeness is defined with respect to the sum of an agent's item-utility and her subsidy.
The idea goes back to \citet{Maskin87} and \citet{Svensson83}, and much early work focuses on assignment-type settings in which each agent receives at most one good~\citep[e.g.,][]{AlkanDeGa91,Aragones95,Klijn00},\footnote{The model of \citet{AlkanDeGa91} allows for undesirable items and negative amounts of money.} with the exception of \citet{MeertensPoRe02}, whose model places no restrictions on the number of agents and goods and allows agents to have general preference relations over their allocated bundle of goods and amount of money.

More recently, \citet{HalpernSh19} initiated the line of research that studies the asymptotic amount of subsidy required to achieve envy-freeness in the setting with $n$ agents, $m$ indivisible goods (with per-item values normalized to $[0, 1]$), and quasi-linear preferences.
In this goods-only setting, \citet{BrustleDiNa20} proved a tight universal bound: one dollar of subsidy per agent always suffices, and this bound is worst-case optimal.
Subsidies have also been studied for \emph{chores} and for fairness notions beyond envy-freeness (e.g., \emph{proportionality} and \emph{maximin share}), including tight and small-subsidy guarantees in several settings~\citep{WuZhZh23,WuZh24,WuXuZh25,WuXuZh25-revisiting}.

Many real-world allocation problems, however, involve a \emph{mix of goods and chores}, and the same indivisible item can be desirable to some agents and undesirable to others.
For instance, in a shared household, one roommate may enjoy cooking while another dislikes it; in a workplace, a late-night shift may be attractive for some employees but undesirable for others; and more generally, responsibilities and perks can be perceived very differently across participants.
This ``mixed manna'' setting has received substantial attention in fair division---via competitive and algorithmic approaches---see~\citep[e.g.,][]{BogomolnaiaMoSa17,KulkarniMeTa21,AzizCaIg22,HosseiniSiVa23,ChaudhuryGaMc23,BarmanHVSe25,BarmanVe26}, and also the survey of \citet{LiuLuSu24} for an overview.

Extending the one-dollar subsidy guarantee from the goods-only setting to such mixed-sign instances is not straightforward.
Existing arguments for goods exploit that every item is weakly beneficial, while analyses for chores rely on all items being weakly harmful; both assumptions fail when an item can be a good for one agent and a chore for another.
Meanwhile, techniques developed for mixed-manna fair division typically do not aim to control the \emph{minimum subsidy} needed for \emph{exact} envy-freeness.
This motivates our central question:
\begin{quote}
\itshape
Given additive utilities with per-item values normalized to $[-1, 1]$, and allowing each item to be a good for some agents and a chore for others, does there always exist an allocation $(A_1, \dots, A_n)$ of the indivisible items among $n$ agents and a subsidy payment vector $p \in [0, 1]^n$ such that
\[
u_i(A_i) + p_i \;\ge\; u_i(A_j) + p_j \qquad \text{for all agents } i, j \; ?
\]
\end{quote}

\subsection{Our Results}

We answer the question posed above affirmatively.
We study the fair division of $m$ indivisible items~$M$ among $n$ agents~$N$ with additive utilities.
As is standard in the subsidy literature, we assume without loss of generality that each marginal value is normalized to lie in $[-1, 1]$ (equivalently, $\max_{i \in N, g \in M} |u_i(g)| \le 1$).
An outcome consists of an allocation $\alloc = (A_1, \dots, A_n)$ and a non-negative subsidy vector $\mathbf{p} = (p_1, \dots, p_n) \in \mathbb{R}_{\ge 0}^n$, where $(\alloc, \mathbf{p})$ is \emph{envy-free} if $u_i(A_i) + p_i \ge u_i(A_j) + p_j$ for all~$i, j \in N$.

\begin{theorem}
\label{thm:main}
For every instance with additive utilities $u_i \colon M \to [-1, 1]$, where each item may be a good for some agents and a chore for others, there exists an envy-free outcome $(\alloc, \mathbf{p})$ such that $0 \le p_i \le 1$ for all~$i \in N$. Furthermore, such an allocation can be computed in polynomial time.
\end{theorem}

The per-agent bound in \Cref{thm:main} is optimal: even in the goods-only setting, some instances require a subsidy of~$1$ for at least one agent~\citep{BrustleDiNa20}.
Thus, one dollar per agent is the best possible worst-case guarantee for the mixed setting as well.
Note that our \Cref{thm:main} recovers the tight subsidy bound for the objective-goods setting~\citep{BrustleDiNa20} and significantly generalizes this result to the more general setting with mixed goods and chores.

En route, in \Cref{sec:structural-properties}, we develop structural properties of iterated maximum-weight matching (used in the work of \citet{BrustleDiNa20}) and an objective-chores analogue
based on iterated maximum-weight \emph{perfect} matchings.
These properties yield clean envy-path bounds not only for objective goods (recovering the result of \citet{BrustleDiNa20} with a much simpler proof), but also for subjective-goods and objective-chores instances.

Given an envy-free outcome $(\alloc, \mathbf{p})$, we can uniformly decrease agents' payments until one agent's payment reaches~$0$.
At that point, at least one agent does not need to receive any subsidy.
Since each of the remaining agents receives a subsidy of at most~$1$, a total subsidy of at most $n-1$ dollars suffices to ensure envy-freeness.

\subsection{Technical Challenges and Overview}
Our proof is guided by the \EFability characterization of
\citet{HalpernSh19}. Once an allocation is fixed, the minimum payments needed
to eliminate envy are determined by longest paths in the corresponding envy
graph. Under our normalization, it therefore suffices to construct an
allocation in which every envy path has weight at most one. The main
difficulty is to enforce this global path bound in the presence of both goods
and chores.

A natural starting point is the recent result of \citet{AzizLuMa26}, who show
that every instance of indivisible mixed manna admits an allocation that is
both \EFOne and \EFable. These two properties alone, however, do not imply a
tight subsidy bound. Even in goods-only instances, an allocation may be both
\EFOne and \EFable while its heaviest envy path has weight \(n-1\). The
reason is that \EFOne controls envy between individual pairs of agents,
whereas the required payments depend on envy accumulated along an entire
path. Thus, the approach of \citet{AzizLuMa26} cannot be applied directly.

We nevertheless adopt one of their central ideas by bundling items into
\emph{meta-goods}. Our construction differs in an important respect, as we
require every meta-good to have value at most one. This quantitative control
is necessary. Without it, even a two-agent instance containing a single
meta-good may require a subsidy strictly greater than one.

Bundling is only the starting point. In the mixed setting, especially after
goods have been bundled with chores, requiring the allocation to remain
\EFOne is too restrictive and is incompatible with the operations used later
in the proof. We therefore do not maintain \EFOne and instead control the
weights of envy paths directly. Our construction distinguishes two cases
according to whether the number of meta-goods exceeds the number of remaining
singleton chores.

The more technically challenging case arises when meta-goods outnumber
singleton chores. Here, we construct the outcome incrementally while
maintaining a collection of carefully chosen invariants. At each intermediate
stage, we work with a subset of the agents and a designated subset of the
special items. The active agents are envy-free under payments bounded by one.
We then add the remaining agents one at a time while preserving envy-freeness
and the subsidy bound within the enlarged active set.

This incremental construction relies on the equality graph and a
payment-reduction procedure. Equality paths identify sequences of
bundle-payment pairs that can be rearranged while preserving the relevant
envy relations. Payment reduction creates the slack needed to incorporate
the next agent without violating the one-dollar bound. Together, these
operations coordinate the allocation of meta-goods and partially allocated
chores with the evolution of the payment vector.

When singleton chores are at least as numerous as the meta-goods, the
additional chores make a different construction possible. We combine a
pairing procedure with thresholding and a flow argument to produce the
desired allocation and payments. Together, the two cases yield an efficient algorithm for computing an envy-free allocation of goods and chores in which each agent receives a subsidy of at most one dollar.

\subsection{Additional Related Work}

\paragraph{Fair Division with Subsidy}
As we mentioned previously, the use of monetary compensation to achieve exact envy-freeness with indivisible items dates back at least to \citet{Maskin87} and \citet{Svensson83} and early algorithmic work on indivisible objects with money \citep[e.g.,][]{Klijn00}.
\citet{HalpernSh19} initiated the modern algorithmic study of subsidies for general additive goods, including an
envy-graph characterization of \EFability and polynomial-time computation of minimum envy-eliminating subsidies
for a fixed allocation.
For goods, \citet{BrustleDiNa20} proved the tight bound of one dollar per agent for additive utilities.
\citeauthor{BrustleDiNa20} also studied the case where agents have general monotone utilities, their subsidy bound was improved by \citet{KawaseMaSuTaYo25}.
Subsequent work refines subsidy bounds and studies additional structure in related models
\citep[e.g.,][]{BarmanKrNa22,GokoIgKa24,LiSuSu25}.

Beyond envy-freeness, subsidy has also been studied for other fairness notions, including proportionality and MMS:
\citet{WuZhZh23} gave tight worst-case total-subsidy bounds for proportionality,
\citet{WuZh24} studied weighted proportionality,
and \citet{WuXuZh25} showed that a small total subsidy suffices for MMS with three agents.
Weighted entitlements in the subsidy model have also been investigated under weighted envy-freeness~\citep{ElmalemAzGo25}.

\citet{NarayanSuVe21} considered the related setting with \emph{transfer payments}, i.e., negative payments are allowed and agents' payments sum to zero.
In this setting, \citeauthor{NarayanSuVe21} studied whether transfer payments can be used to achieve both envy-freeness and economic efficiency.

Most of this literature focuses on instances consisting entirely of goods or entirely of chores; in contrast, we study envy-freeness with subsidies in the fully mixed setting where each item may be a good for some agents and a chore for others.
The work of \citet{WuXuZh25-revisiting} also concerned the same setting and showed that it suffices to guarantee proportionality with a total subsidy of~$\frac{n}{4}$ when $n$ is even and $\frac{n^2 - 1}{4n}$ when $n$ is odd; both bounds are tight~\citep{WuZhZh23}.

\paragraph{Rent Division}
Rent division allocates indivisible rooms together with payments.
Under quasilinear preferences and budget-balance constraints, envy-free outcomes are guaranteed to exist
\citep[e.g.,][]{Su99}.
This differs from subsidy-based fair division, where payments are required to be nonnegative subsidies funded
by an external source, leading to different feasibility and extremal questions.

\paragraph{Relations of Envy-Freeness}
When exact envy-freeness is unattainable without money, a large literature studies relaxations such as
\EFOne and \EFX, as well as share-based notions
such as MMS~\citep{Budish11,CaragiannisKuMoPrShWa19,LiptonMaMoSa04}.
We refer to the following survey articles~\citep{AmanatidisAzBi23,GuoLiDe23,LiuLuSu24,NguyenRo23,Suksompong21,Suksompong25} for an overview of recent developments and progress.
This line of work is complementary to ours: we keep exact envy-freeness, but allow bounded subsidies to restore
feasibility.

\section{Preliminaries}
\label{sec:preliminaries}

For any positive integer~$s$, let $[s] \coloneqq \{1, 2, \dots, s\}$.
A fair division instance is a tuple $\langle N, M, (u_i)_{i \in N} \rangle$, where
$N = [n]$ is the set of agents and $M = [m]$ is the set of indivisible items.
Each agent~$i\in N$ has an \emph{additive utility function} $u_i \colon 2^M \to \mathbb{R}$, determined by item values $u_i(j) \in \mathbb{R}$ for $j \in M$; i.e.,
\[
u_i(S) \;=\; \sum_{j \in S} u_i(j) \qquad \forall S \subseteq M.
\]
We assume the standard normalization that
\[
u_i(j) \in [-1, 1] \qquad \forall i \in N,\ \forall j \in M.
\]
For our computational claims, the item values are rational numbers encoded in binary, and running time is measured in the total input bit length.

An item~$j \in M$ is a \emph{subjective good} if $u_i(j) \ge 0$ for some~$i \in N$.
An item~$j \in M$ is an \emph{objective good} (resp., \emph{objective chore}) if $u_i(j) \ge 0$ (resp., $u_i(j) < 0$) for all agents~$i \in N$.

An \emph{allocation} is an ordered partition $\alloc = (A_1, \dots, A_n)$ of~$M$ into $n$ (possibly empty) bundles, where agent~$i$ receives bundle~$A_i$.
The allocation~$\alloc$ is \emph{envy-free} if for all~$i, j \in N$,
\[
u_i(A_i)\ \ge\ u_i(A_j).
\]

A \emph{subsidy vector} is $\mathbf{p} = (p_1, \dots, p_n) \in \mathbb{R}_{\ge 0}^n$, where $p_i$ is a payment received by agent~$i$ from an external sponsor.
An \emph{allocation with payments} $(\alloc, \mathbf{p})$ is \emph{envy-free (EF)} if for all~$i, j \in N$,
\[
u_i(A_i) + p_i\ \ge\ u_i(A_j) + p_j.
\]
We refer to $\sum_{i\in N} p_i$ as the \emph{total subsidy}. Our objective is to find an envy-free outcome $(\alloc, \mathbf{p})$ minimizing total subsidy.

\subsection{Envy-Freeability and the Envy Graph}
\label{subsec:EFable-envy-graph}

An allocation~$\alloc$ is \emph{\EFable} if there exists a subsidy vector $\mathbf{p} \in \mathbb{R}_{\ge 0}^n$ such that $(\alloc, \mathbf{p})$ is envy-free.

\begin{definition}
\label{def:envy-graph}
Given an allocation $\alloc = (A_1, \dots, A_n)$, its \emph{envy graph} is the complete directed graph $\mathcal{G}_\alloc$ on vertex set~$N$ with edge weights
\[
w_\alloc(i, j)\ \coloneqq\ u_i(A_j) - u_i(A_i) \qquad \forall i, j \in N.
\]
A (directed) path $P = (i_1, i_2, \dots, i_k)$ has weight $w_\alloc(P) \coloneqq \sum_{t = 1}^{k-1} w_\alloc(i_t, i_{t+1})$.
For $i, j \in N$, let $\ell_\alloc(i, j)$ denote the maximum weight of any directed path from~$i$ to~$j$ in $\mathcal{G}_\alloc$, and let
\[
\ell_\alloc(i)\ \coloneqq\ \max_{j \in N} \ell_\alloc(i, j)
\]
denote the maximum weight of any directed path starting at~$i$.
(We allow the length-$0$ path, so $\ell_\alloc(i) \ge 0$ for all~$i$.)
\end{definition}

\citet{HalpernSh19} gave a characterization of \EFability for goods; however, the characterization also works for the setting with mixed goods and chores.

\begin{theorem}[\citealp{HalpernSh19}]
\label{thm:Halpern-Shah-EFable}
For an allocation~$\alloc$, the following are equivalent:
\begin{enumerate}[label=(\alph*)]
\item $\alloc$ is \EFable.
\item $\alloc$ maximizes the utilitarian welfare across all reassignments of its bundles; i.e., for every permutation~$\sigma$ of~$[n]$,
\[
\sum_{i \in N} u_i(A_i)\ \ge\ \sum_{i \in N} u_i(A_{\sigma(i)}).
\]
\item The envy graph~$\mathcal{G}_\alloc$ has no positive-weight directed cycle.
\end{enumerate}
\end{theorem}

\paragraph{Optimal Subsidies for a Fixed Envy-Freeable Allocation}
When $\alloc$ is \EFable (equivalently, $\mathcal{G}_\alloc$ has no positive cycle),
\citet{HalpernSh19} show that the payment vector~$\mathbf{p}$ defined by
\[
p_i\ \coloneqq\ \ell_\alloc(i) \qquad \forall i \in N
\]
makes $(\alloc, \mathbf{p})$ envy-free; moreover, it is componentwise minimal among all non-negative payment vectors that make~$\alloc$ envy-free, and thus minimizes total subsidy among those payment vectors.

\subsection{Iterated Maximum-Weight Matching Procedures}
\label{subsec:IMWM-IMWPM}

We use two iterative matching procedures as subroutines.
In both procedures, the input consists of a set of agents~$N$ and a set of indivisible \emph{objects}~$J$ (items, meta-goods, meta-chores, and/or dummy objects).
Each object~$o \in J$ has a well-defined utility $u_i(o)$ for each agent~$i \in N$.
Both procedures return an allocation $\alloc = (A_1, \dots, A_n)$ of (a superset of) $J$ into~$n$ bundles, where each~$A_i$ is the multiset of objects assigned to agent~$i$ across rounds.\footnote{When dummy objects are used, they can be ignored in the induced allocation of the original objects.}
Below, we define \emph{bipartite matching graphs}: $G[N, J]$ and $H[N, J]$, which are unrelated to the envy graph $\mathcal{G}_\alloc$ in \Cref{def:envy-graph}.

\begin{definition}[Matching graphs]
\label{def:matching-graphs}
Fix a set of agents~$N$ and a set of indivisible \emph{objects}~$J$, together with utilities $(u_i)_{i \in N}$ over objects.

\begin{itemize}
\item \textbf{Goods-style graph.}
Let $G[N, J] = (N \cup J, E_G)$ be the \emph{complete} weighted bipartite graph where $E_G \coloneqq N \times J$,
and each edge $(i, o) \in E_G$ has weight $u_i(o)$.
For any $J' \subseteq J$, let $G[N, J']$ denote the subgraph of $G[N, J]$ induced by $N \cup J'$.

A \emph{maximum-weight matching} in $G[N, J']$ is a matching maximizing total edge weight, breaking ties in favour of maximum cardinality among all maximum-weight matchings.

\item \textbf{Chores-style graph.}
Let $H[N, J] = (N \cup J, E_H)$ be the \emph{complete} weighted bipartite graph where $E_H \coloneqq N \times J$, and each edge $(i, o) \in E_H$ has weight $u_i(o)$.
For any $J' \subseteq J$, let $H[N, J']$ denote the subgraph of $H[N, J]$ induced by $N \cup J'$.

A \emph{maximum-weight perfect matching} in $H[N, J']$ is a maximum-weight matching of size $|N|$ that matches every agent in~$N$ to exactly one distinct object in~$J'$ (it may leave objects unmatched when $|J'| > |N|$).
\end{itemize}
\end{definition}

Next, we introduce the two iterative matching procedures: \IMWM (Iterated Maximum-Weight Matching) and \IMWPM (Iterated Maximum-Weight Perfect Matching); see also \Cref{alg:IMWM,alg:IMWPM} for the pseudocode.

\begin{algorithm}[t]
\caption{\IMWM (Iterated Maximum-Weight Matching)}
\label{alg:IMWM}
\DontPrintSemicolon

\KwIn{Complete weighted bipartite graph $G[N, J]$ as in \Cref{def:matching-graphs}, where every object has a non-negative edge to some agent.}
\KwOut{Allocation $\alloc = (A_1, \dots, A_n)$.}

$A_i \gets \emptyset$ for all~$i \in N$.\;
$t \gets 1$; $J^1 \gets J$.\;
\While{$J^t \neq \emptyset$}{
	Compute a maximum-weight matching $\mu^t = \{(i, \mu_i^t)\}$ in $G[N, J^t]$.\;
	\lForEach{$i \in N$}{
			$A_i \gets A_i \cup \{\mu_i^t\}$ if $\mu_i^t \neq \emptyset$.
	}
	$J^{t+1} \gets J^t \setminus \{\mu_i^t \colon i \in N,\ \mu_i^t \neq \emptyset\}$.\;
	$t \gets t+1$.\;
}

\Return{$(A_1, \dots, A_n)$}
\end{algorithm}

\begin{definition}[\IMWM]
\label{def:IMWM}
Given a goods-style matching graph $G[N, J]$ (see Definition~\ref{def:matching-graphs}) in which every object has a non-negative edge to some agent, \IMWM (Algorithm~\ref{alg:IMWM}) allocates the objects in \emph{rounds}.
It maintains a set $J^t$ of objects not yet allocated.
In each round~$t$, it computes a maximum-weight matching in the induced subgraph $G[N, J^t]$, assigns each matched agent the (at most one) object she is matched to, removes all assigned objects from~$J^t$, and continues until no objects remain.

The output allocation $\alloc = (A_1, \dots, A_n)$ collects, for each agent~$i$, all objects assigned to~$i$ across rounds (possibly none).
Although $G$ is complete, a maximum-weight matching never contains a negative-weight edge, since deleting such an edge would strictly increase its weight. Hence \IMWM never assigns an object to an agent who values it negatively.
\end{definition}

\begin{algorithm}[t]
\caption{\IMWPM (Iterated Maximum-Weight Perfect Matching)}
\label{alg:IMWPM}
\DontPrintSemicolon

\KwIn{Complete weighted bipartite graph $H[N, J]$ as in \Cref{def:matching-graphs}.}
\KwOut{Allocation $\alloc = (A_1, \dots, A_n)$.}

If $|J|$ is not a multiple of~$n$, add dummy objects of value~$0$ to all agents so that $|J| = T \cdot n$ for some~$T \in \mathbb{N}$.\;
$A_i \gets \emptyset$ for all~$i \in N$.\;
$t \gets 1$; $J^1 \gets J$.\;
\While{$J^t \neq \emptyset$}{
	Compute a maximum-weight perfect matching $\mu^t = \{(i, \mu_i^t)\}$ in $H[N, J^t]$.\;
	\lForEach{$i \in N$}{
		$A_i \gets A_i \cup \{\mu_i^t\}$.
	}
	$J^{t+1} \gets J^t \setminus \{\mu_i^t \colon i \in N\}$.\;
	$t \gets t+1$.\;
}

\Return{$(A_1, \dots, A_n)$}
\end{algorithm}

\begin{definition}[\IMWPM]
\label{def:IMWPM}
Given a chores-style matching graph $H[N, J]$ (see Definition~\ref{def:matching-graphs}), \IMWPM (Algorithm~\ref{alg:IMWPM}) is the analogous iterated procedure based on \emph{maximum-weight perfect matchings}.
If $|J|$ is not a multiple of~$n$, it first adds dummy objects of value~$0$ for all agents so that $|J| = T \cdot n$ for some integer~$T$.
It then proceeds in rounds with a remaining set $J^t$: in each round~$t$, it computes a maximum-weight perfect matching in $H[N, J^t]$, assigns every agent exactly one object in that matching (possibly a dummy), removes the assigned objects from~$J^t$, and repeats until $J^t = \emptyset$.

The output allocation $\alloc = (A_1, \dots, A_n)$ is obtained by collecting each agent's assigned objects across rounds (ignoring dummy objects).
\end{definition}

We will use that the allocations returned by \IMWM and \IMWPM are \EFable; see the proof of \citet{BrustleDiNa20} for goods and the analogous argument for objective chores~\citep[e.g.,][]{AzizLuMa26}.

\subsection{Meta-Goods and Chore-Maximality}
\label{subsec:metagoods}

Finally, we introduce some important ingredients that are needed for our proof.

\begin{definition}[Meta-good]
A non-empty subset $M'\subseteq M$ is a \emph{meta-good} if $u_i(M')\ge 0$ for some agent $i\in N$.
\end{definition}

Note that a meta-good may be a singleton.
In our algorithms and arguments, we treat each meta-good as an individual item rather than as a set of items.
A natural restriction on meta-goods, when reasoning about subsidies is the following property:

\begin{definition}[Chore-maximality]
A non-empty subset $M '\subseteq M$ is \emph{chore-maximal} if it is a meta-good and, for every objective chore $c \in M \setminus M'$, we have $u_i(M' \cup \{c\}) < 0$ for all~$i \in N$.
\end{definition}

\section{Structural Properties of Iterated Maximum-Weight Matching}
\label{sec:structural-properties}

In this section, we introduce the key properties of \IMWM that are used throughout the paper.
Our structural insights not only provide a simpler proof of the result of \citet{BrustleDiNa20}, which shows that one dollar each suffices for objective-goods-only instances, but also allow for a generalization to instances with only subjective goods or only objective chores.

\subsection{Subjective-Goods-Only Instances}

It is worth noting that in this special case, we only require that each agent's utility for each item is upper bounded by~$1$, deviating from the common normalization of~$[-1, 1]$ stated in \Cref{sec:preliminaries}.
The generality of this result allows us to apply it directly to a subcase of the mixed-goods-and-chores setting in \Cref{sec:main:0-remaing-chores}.

\begin{proposition}
\label{prop:SG_only}
Suppose that for every item~$g \in M$ and every agent~$i \in N $, the utility satisfies $u_i(g) \in (-\infty, 1]$, and that for each item~$g$, there exists at least one agent with non-negative utility for~$g$.
Then, \IMWM produces an envy-free allocation $(\alloc, \mathbf{p})$ such that $0 \le p_i \le 1$ for all~$i \in N$.
\end{proposition}

\begin{proof}
Let $\alloc = (A_1, \dots, A_n)$ be the allocation returned by \IMWM, and let $J^t$ be the set of items available at the beginning of round~$t$.
For each agent~$i$, let $\mu_i^t$ be the item assigned to~$i$ in round~$t$; if~$i$ is unmatched, set $\mu_i^t = \emptyset$ and $u_i(\emptyset)=0$.
As observed above, no edge used by a maximum-weight matching has negative weight.
Moreover, the maximum-cardinality tie-break ensures that a nonempty set~$J^t$ containing a subjective good always yields a nonempty matching, so every item is eventually allocated.

We first show that~$\alloc$ is \EFable.
Fix a round~$t$ and a permutation~$\sigma$ of~$N$.
Because the goods-style graph is complete, reassigning $\mu_{\sigma(i)}^t$ to agent~$i$ for every~$i$ gives a feasible matching after empty assignments are ignored.
The optimality of~$\mu^t$ therefore implies
\[
\sum_{i \in N} u_i(\mu_i^t)
\ge
\sum_{i \in N} u_i(\mu_{\sigma(i)}^t).
\]
Summing over all rounds gives
\[
\sum_{i \in N}u_i(A_i)
\ge
\sum_{i \in N}u_i(A_{\sigma(i)}).
\]
By \Cref{thm:Halpern-Shah-EFable}, the allocation~$\alloc$ is \EFable.

It remains to bound the weight of every path in its envy graph.
Fix a directed path $P=(1,2,\dots,k)$ and let
\[
w_t(P) \coloneqq \sum_{i=1}^{k-1}
\left(u_i(\mu_{i+1}^t)-u_i(\mu_i^t)\right)
\]
be the contribution of round~$t$ to its weight.
For the terminal agent~$k$, define
\[
F_t \coloneqq \max\bigl(\{0\}\cup\{u_k(g):g\in J^t\}\bigr),
\]
and set $F_{T+1}=0$ after the final round.

In round~$t$, consider the alternative matching that assigns $\mu_{i+1}^t$ to agent~$i$ for each $i\in\{1,\dots,k-1\}$, leaves the original assignments of agents outside the path unchanged, and gives agent~$k$ an item in~$J^{t+1}$ of value~$F_{t+1}$ if this value is positive, leaving~$k$ unmatched otherwise.
This matching is feasible because the graph is complete and the additional item, when used, remains unallocated after round~$t$.
Maximum-weight optimality gives
\[
F_{t+1}-u_k(\mu_k^t)+w_t(P)\le 0.
\]
Since $u_k(\mu_k^t)\le F_t$, we obtain
\[
w_t(P)\le F_t-F_{t+1}.
\]
Summing over the rounds yields
\[
w_\alloc(P)=\sum_{t=1}^T w_t(P)
\le \sum_{t=1}^T(F_t-F_{t+1})
=F_1\le 1.
\]
Thus every directed path has weight at most~$1$.
The longest-path payment characterization of \citet{HalpernSh19} now gives an envy-free outcome with $0\le p_i\le 1$ for every agent~$i$.
\end{proof}

\subsection{Objective-Chores-Only Instances}

The proof for objective-chores only instances is similar to that of \Cref{prop:SG_only}, we thus only state the result below and defer its proof to  Appendix~\ref{app:prop32}.

\begin{proposition}
\label{prop:OC_only}
Suppose that every item~$g \in M$ is an objective chore and that, for every agent~$i \in N$, the utility satisfies $u_i(g) \geq -1$.
Then, \IMWPM produces an envy-free allocation $(\alloc, \mathbf{p})$ such that $0 \le p_i \le 1$ for all~$i \in N$.
\end{proposition}

\section{Proof of the Main Result}

In this section, we prove \Cref{thm:main}, the main result of the paper, which states that as long as the marginal utility or the marginal disutility of each item is bounded above by one, there exists an envy-free allocation in which each agent is subsidized by at most one.
Since the main focus of this section is the existence result, we will defer the polynomial-time implementation of certain steps to \Cref{sec:polytime}.

Previously in \Cref{sec:structural-properties}, we exploited the structures of subjective-goods-only or objective-chores-only instances to guide the item allocation.
In this section, however, we no longer make any assumption about the structure of the set of agents with positive or negative utilities for an item.
As a result, no analogous structure is available in this general case.
Thus, the problem becomes substantially more difficult.

We now describe the high-level idea of our algorithm and proof.
We begin with an initial bundling that partitions the set of items~$M$ into objective chores (denoted as~$Z_{\mathrm{rem}}$) and subjective meta-goods (denoted as~$G$), each of which satisfies certain structural properties presented in the next subsection.
We then branch into three cases: (i) $|Z_{\mathrm{rem}}| = 0$, (ii) $0 < |Z_{\mathrm{rem}}| < |G|$, and (iii) $|Z_{\mathrm{rem}}| \geq |G|$.
The first case is relatively straightforward, as we can apply \IMWM (\Cref{alg:IMWM}) directly to this instance without any rebundling.
The second case is the most technically involved, as it requires both rebundling steps and a more intricate allocation procedure and analysis.
The final case can be transformed into an objective-chores-only instance, in which, with appropriate bundling, we can apply \IMWPM (\Cref{alg:IMWPM}).
We show that in each of these cases, at most one dollar per agent suffices to achieve envy-freeness, thereby establishing Theorem~\ref{thm:main}.

\subsection{Initial Bundling of Items}

\begin{algorithm}[t]
\caption{Conditional Pairwise Merging}
\label{alg:condpairwise}
\DontPrintSemicolon

\KwIn{Agents~$N$, indivisible items~$M$ (partitioned into subjective goods $U_{\mathrm{init}}$ and objective chores $Z_{\mathrm{init}}$), and agents' utilities~$(u_i)_{i \in N}$.}
\KwOut{A set of meta-goods $G = \{M_1, \dots, M_\ell\}$ and a set of residual objective chores~$Z_{\mathrm{rem}}$.}

Initialize $U \gets U_{\mathrm{init}}$ and $Z \gets Z_{\mathrm{init}}$.\;

\BlankLine
\tcp{Phase 1: Greedy Cancellation (Individual)}
\While(\tcp*[f]{Pre-process single items.}){$\exists g \in U, c \in Z, i \in N$ such that $u_i(g \cup \{c\}) \geq 0$}{
	$g \gets g \cup \{c\}$ \tcp*{Merge chore into good (Value $\le 1$).}
	$Z \gets Z \setminus \{c\}$\;
}

\BlankLine
\tcp{Phase 2: Pairwise Merging with Chore Maximality Check}
Construct the interest graph $R[N,U]=(N\cup U,E)$, where $(i,g)\in E$ if and only if $u_i(g)\geq 0$.\;
\While{$|Z| > 0$ \textup{\bfseries and} $\exists i \in N$ with $\deg(i) \geq 2$ in $R[N, U]$}{ \label{alg:condpairwise:loop}
	Select agent~$i \in N$ with $\deg(i) \geq 2$.\;
	Select two distinct items $g_1, g_2 \in \{g \in U \mid (i, g) \in E\}$.\;
	$S_{\mathrm{temp}} \gets g_1 \cup g_2$ \tcp*{Tentative merge.}

	\tcp{Check if new bundle is ``Too Strong'' (Not Chore-Maximal)}
	\eIf{$\exists c \in Z$ such that $u_j(S_{\mathrm{temp}} \cup \{c\}) \geq 0$ for some~$j \in N$}{ \label{alg:IfMerge}
		$g_{\mathrm{new}} \gets S_{\mathrm{temp}} \cup \{c\}$ \tcp*{Consume chore, cancelling out one of the goods.}
		$Z \gets Z \setminus \{c\}$\;
	}{
		$g_{\mathrm{new}} \gets S_{\mathrm{temp}}$ \label{alg:ElseMerge} \tcp*{Already maximal (dominated by chore $\le 1$).}
	}

	\lForEach{agent~$j \in N$}{
		Add an edge $(j, g_{\mathrm{new}})$ to~$E$ if $u_j(g_{\mathrm{new}}) \geq 0$.
	}
	Remove edges incident to $g_1, g_2$ from~$E$.\;
	Update $U \gets (U \setminus \{g_1, g_2\}) \cup \{g_{\mathrm{new}}\}$.\;
}

\Return{A set of meta-goods $G \gets U$ and residual chores~$Z_{\mathrm{rem}} \gets Z$}
\end{algorithm}

Our initial bundling of the items aims to partition~$M$ into a set of residual \emph{objective chores} $Z_{\mathrm{rem}}$ and a family of \emph{meta-goods} $G = \{M_1, \dots, M_\ell\}$, where each meta-good satisfies \choreMaximality with respect to $Z_{\mathrm{rem}}$.
When $Z_{\mathrm{rem}} \neq \emptyset$, we additionally enforce that every agent is interested in \emph{at most one} meta-good (equivalently, the interest sets of the meta-goods are pairwise disjoint).

\begin{definition}[Interest set]
\label{def:interest-set}
For any meta-good~$M_j \in \{M_1, \dots, M_\ell\}$, we define the \emph{interest set}~$T_j$ as the set of agents who derive non-negative utility from it:
\[
T_j \coloneqq \{i \in N \mid u_i(M_j) \geq 0\}.
\]
\end{definition}

We will build a bundling where the interest sets are disjoint via \Cref{alg:condpairwise}.
The algorithm is closely related to the iterative item-merging/bundling routine used in the study of envy-freeness for mixed resources (EFM)~\citep{AzizLuMa26}.
At a high level, both procedures repeatedly merge items into \emph{meta-goods} so as to separate agents' interest structure and to enforce a form of chore-maximality with respect to the remaining objective chores.
The key additional constraint in the subsidy setting is that every resulting meta-good must remain \emph{unit-bounded} for all interested agents, i.e., $u_i(M_j) \le 1$ for all~$i \in T_j$ (Property~\ref{prop:PTwo} in Lemma~\ref{lem:algoOneProperties}).
This upper bound is crucial later when we treat each meta-good as a single $[0, 1]$-valued object in the \IMWM analysis and bound envy-graph path weights.
Accordingly, Algorithm~\ref{alg:condpairwise} consumes objective chores during merging so as to preserve chore-maximality w.r.t.\ $Z_{\mathrm{rem}}$, which in turn implies the desired upper bound.

We now describe Algorithm~\ref{alg:condpairwise} in more detail.
Formally, let
\[
U_{\mathrm{init}} \coloneqq \{g \in M \colon \exists i \in N \text{ with } u_i(g) \ge 0\}
\qquad\text{and}\qquad
Z_{\mathrm{init}} \coloneqq M \setminus U_{\mathrm{init}}.
\]
Note that $Z_{\mathrm{init}}$ is exactly the set of \emph{objective chores} (items that are negative for all agents), and $U_{\mathrm{init}}$ contains all subjective goods (including objective goods).
Algorithm~\ref{alg:condpairwise} maintains a current set~$U$ of (meta-)goods (each element of~$U$ is a subset of original items) and a current set~$Z$ of remaining objective chores, initialized as $U \gets U_{\mathrm{init}}$ and $Z \gets Z_{\mathrm{init}}$.

A key book-keeping device is the bipartite \emph{interest graph} $R[N, U] = (N \cup U, E)$, where an edge $(i, g) \in E$ indicates that agent~$i$ values the (meta-)good $g \in U$ non-negatively, i.e., $u_i(g) \ge 0$.\footnote{The edge weights can be taken as $u_i(g)$, though Algorithm~\ref{alg:condpairwise} uses only adjacency and degrees.}
For~$i \in N$, let $\deg(i)$ denote the degree of~$i$ in~$R[N, U]$.

Algorithm~\ref{alg:condpairwise} has the following two phases.
\begin{itemize}
\item \textbf{Phase 1 (Greedy cancellation).}
As long as there exist $g \in U$, $c \in Z$, and an agent~$i \in N$ with $u_i(g \cup \{c\}) \ge 0$, we merge the chore into the good by setting $g \gets g \cup \{c\}$ and deleting~$c$ from~$Z$.
When this phase terminates, we obtain the invariant that for every remaining $g \in U$ and every remaining~$c \in Z$,
\[
u_i(g \cup \{c\}) \;=\; u_i(g) + u_i(c) \;<\; 0 \qquad \forall i \in N,
\]
i.e., no remaining objective chore can be added to any current (meta-)good without making the resulting bundle strictly negative for every agent.
This is the core ``chore-maximality'' condition we will use later.

\item \textbf{Phase 2 (Pairwise merging to enforce disjoint interest sets).}
While $Z \neq \emptyset$ and there exists an agent~$i$ with $\deg(i) \ge 2$ in $R[N, U]$, we pick two distinct neighbors $g_1, g_2 \in U$ of~$i$ and tentatively merge them into $S_{\mathrm{temp}} \gets g_1 \cup g_2$.
If this tentative merge is ``too strong'' in the sense that there exists a remaining chore~$c \in Z$ and an agent~$j\in N$ with $u_j(S_{\mathrm{temp}} \cup \{c\}) \ge 0$, we consume one such chore by setting $g_{\mathrm{new}} \gets S_{\mathrm{temp}} \cup \{c\}$ and removing $c$ from $Z$; otherwise we keep $g_{\mathrm{new}} \gets S_{\mathrm{temp}}$.
We then replace $g_1$ and $g_2$ in~$U$ by $g_{\mathrm{new}}$ and update the interest graph accordingly.
Each iteration strictly decreases~$|U|$, so Phase~2 terminates.
\end{itemize}

We summarize below in \Cref{lem:algoOneProperties} the properties of the output of Algorithm~\ref{alg:condpairwise} and defer its proof to Appendix~\ref{App:Algone}.

\begin{lemma}
\label{lem:algoOneProperties}
Let $G = \{M_1, \dots, M_\ell\}$ be the set of meta-goods and $Z_{\mathrm{rem}}$ the set of objective chores returned by \Cref{alg:condpairwise}.
Then, the following properties hold:
\begin{enumerate}[label=(P\arabic*),ref=(P\arabic*)]
\item\label{prop:POne}
\textbf{Disjoint interest sets.}
Either $Z_{\mathrm{rem}} = \emptyset$, or for any pair of meta-goods $M_i, M_j \in G$, their interest sets are disjoint, i.e., $T_i \cap T_j = \varnothing$ for all $i \neq j$.

\item\label{prop:PTwo}
\textbf{Upper bound on interested values.}
For each~$M_j \in G$ and each~$i \in T_j$, $0 \le u_i(M_j) \le 1$.

\item\label{prop:PThree}
\textbf{Chore-maximality w.r.t.\ residual chores.}
For each~$M_j \in G$, each $c \in Z_{\mathrm{rem}}$, and each agent~$i \in N$,
\[
u_i(M_j \cup \{c\}) = u_i(M_j) + u_i(c) < 0.
\]
\end{enumerate}
\end{lemma}

\subsection[Case I: Running Out of Chores]{Case I: Running Out of Chores ($|Z_{\mathrm{rem}}| = 0$)}
\label{sec:main:0-remaing-chores}

When $|Z_{\mathrm{rem}}| = 0$, Algorithm~\ref{alg:condpairwise} outputs only meta-goods $G = \{M_1, \dots, M_\ell\}$ (and no residual chores).
We therefore obtain an induced instance on the object set~$G$, where each meta-good~$M_j$ is treated as a single indivisible object with value $u_i(M_j)$ for agent~$i$ (well-defined by additivity).

We now verify that the hypotheses of \Cref{prop:SG_only} hold for this induced instance.
By \Cref{lem:algoOneProperties}, each meta-good~$M_j$ has a nonempty interest set $T_j = \{i \in N \colon u_i(M_j) \ge 0\}$; hence every object~$M_j$ is a subjective good (some agent values it non-negatively).
Moreover, for any~$M_j \in G$ and any agent~$i \in N$, we have $u_i(M_j) \le 1$: if $u_i(M_j) \ge 0$, then $i \in T_j$ and Property~\ref{prop:PTwo} of \Cref{lem:algoOneProperties} gives $u_i(M_j) \le 1$, while if $u_i(M_j) < 0$, then the inequality is trivial.
Thus, the induced instance satisfies the assumptions of \Cref{prop:SG_only}.

Therefore, running \IMWM on the meta-goods $G$ yields an allocation whose maximum
envy-eliminating subsidy is at most~$1$ per agent, completing Case~I.

\subsection[Case II: Sparse Chores]{Case II: Sparse Chores ($0<|Z_{\mathrm{rem}}|<|G|$)}
\label{subsec:sparse-chores}

We now turn to the algorithm and analysis of the case in which the number of residual objective chores~$Z_{\mathrm{rem}}$ is strictly less than the number of created meta-goods $G = \{M_1, \dots, M_\ell\}$ when \Cref{alg:condpairwise} terminates.
Recall that each~$c \in Z_{\mathrm{rem}}$ is an objective chore with $u_i(c) \in [-1, 0)$ for all agents~$i \in N$.
For ease of notation, let $k \coloneqq |Z_{\mathrm{rem}}|$.
We thus have $0 < k < \ell$.

\paragraph{Ordering and Relabeling Agents and Meta-Goods}
In order to better describe our algorithm, we first relabel agents and meta-goods.
For each meta-good~$M_j \in G$, let $v_j = \max_{i \in T_j} u_i (M_j)$ be the maximum value an agent has for~$M_j$.
Relabel the meta-goods so that
$$
1 \geq v_1 \geq v_2 \geq \cdots \geq v_k \geq v_{k+1} \geq \cdots \geq v_\ell \geq 0.
$$
Next, for each meta-good~$M_j$, relabel~$j \in N$ as the agent attaining this maximum, with arbitrary tie-breaking, i.e.,
$$
u_j(M_j) = v_j \qquad \forall j \in [\ell].
$$
Note that each such agent is from the interest set~$T_j$.
This relabelling is well-defined because these agents are all distinct: by Property~\ref{prop:POne} of \Cref{lem:algoOneProperties}, the interest sets $T_i$ and $T_j$ are disjoint for any $i \neq j$.

\begin{algorithm}[p]
\caption{Envy-Free Outcome for the Case of Sparse Chores}
\label{alg:sparse}
\DontPrintSemicolon
\small

\KwIn{Agents~$N$, meta-goods $G=\{M_1,\dots,M_\ell\}$, objective chores $Z_{\mathrm{rem}}=\{c_1,\dots,c_k\}$, and utilities $(u_i)_{i\in N}$.}
\KwOut{An envy-free allocation with payments $(\alloc,\mathbf p)$.}

\lForEach{$M_j\in G$}{$v_j\gets\max_{i\in T_j}u_i(M_j)$.}
Relabel the meta-goods so that $1\ge v_1\ge\cdots\ge v_\ell\ge0$, and relabel the agents so that $u_j(M_j)=v_j$ for every $j\in[\ell]$.\;

\BlankLine
\tcp{Construct the partial outcome for $\{M_1,\dots,M_k\}\cup Z_{\mathrm{rem}}$.}
$\lambda\gets v_{k+1}$; $A_i\gets\emptyset$ and $p_i\gets0$ for all $i\in N$; $\widetilde N\gets[k]$.\;
Allocate $Z_{\mathrm{rem}}$ to $\widetilde N$ using \IMWPM; let $c_i^*$ be the chore assigned to each $i\in\widetilde N$, and set $A_i\gets\{c_i^*\}$.\;
Compute, using \Cref{prop:OC_only}, a payment vector $\mathbf p^0$ that makes the chore allocation envy-free on $\widetilde N$.\;
$p^0_{\max}\gets\max_{i\in\widetilde N}p_i^0$.\;
\lForEach{$i\in\widetilde N$}{$p_i\gets p_i^0+(1-p^0_{\max})$; $A_i\gets A_i\cup M_i$.}
Replace $(p_i)_{i\in\widetilde N}$ by the payment vector returned by \Cref{lem:RestoreI2} for $(\alloc,\mathbf p)$ on $\widetilde N$.\;

\BlankLine
\tcp{Add the remaining agents one at a time.}
\While{$\widetilde N\neq N$}{
    Pick any $x\in N\setminus\widetilde N$.\;
    $\widetilde N^+\gets\widetilde N\cup\{x\}$; $(\alloc^+,\mathbf p^+)\gets(\alloc,\mathbf p)$; $(A_x^+,p_x^+)\gets(\emptyset,\lambda)$.\;
    \If{$u_x(A_x^+)+p_x^+<\max_{j\in\widetilde N}\bigl(u_x(A_j)+p_j\bigr)$}{
        Choose $i_1\in\argmax_{j\in\widetilde N}\bigl(u_x(A_j)+p_j\bigr)$.\;
        Choose a directed equality path $i_1\to i_2\to\cdots\to i_s$ in $\mathcal H(\alloc,\mathbf p)$ with $i_s\in S(\lambda,\mathbf p)$.\;
        $(A_x^+,p_x^+)\gets(A_{i_1},p_{i_1})$.\;
        \lForEach{$r\in\{1,\dots,s-1\}$}{$(A_{i_r}^+,p_{i_r}^+)\gets(A_{i_{r+1}},p_{i_{r+1}})$.}
        $(A_{i_s}^+,p_{i_s}^+)\gets(\emptyset,\lambda)$.\;
        Replace $(p_i^+)_{i\in\widetilde N^+}$ by the payment vector returned by \Cref{lem:RestoreI2} for $(\alloc^+,\mathbf p^+)$ on $\widetilde N^+$.\;
    }
    $(\widetilde N,\alloc,\mathbf p)\gets(\widetilde N^+,\alloc^+,\mathbf p^+)$ \tcp*{Commit the candidate state.}
}

\BlankLine
\tcp{Allocate the remaining meta-goods.}
\lForEach{$j\in\{k+1,\dots,\ell\}$}{$A_j\gets A_j\cup M_j$; $p_j\gets p_j-v_j$.}

\Return{$(\alloc,\mathbf p)$}\;
\end{algorithm}

\paragraph{Overview of the Algorithm}
We are now ready to present the high-level idea of our algorithm, which consists of two stages; the full procedure appears in \Cref{alg:sparse}.
\begin{itemize}
\item In the first stage of \Cref{alg:sparse}, we allocate the first $k$ meta-goods $\{M_j\}_{j \in [k]}$ and all residual objective chores $Z_{\mathrm{rem}} = \{c_1, \dots, c_k\}$.
Allocating these items constitutes the main obstacle to allocating all items while maintaining bounds on individual subsidies.

Since the number of objective chores is less than the number of meta-goods, there must exist some agent~$i$ who gets her meta-good (assuming~$M_j$) without being allocated to any objective chores from~$Z_{\mathrm{rem}}$.
Envy-freeness between agents in the interest set~$T_j$ suggests that each agent in~$T_j \setminus \{i\}$ must receive a positive subsidy payment to eliminate their envy towards agent~$i$.
Next, in order to ensure envy-freeness from agents outside~$T_j$ to agents within~$T_j$, we need to carefully allocate the meta-goods, the objective chores, and subsidy payments.

Our goal at this stage of the algorithm is to obtain an allocation $\alloc$ of $\{M_j\}_{j \in [k]}$ and $Z_{\mathrm{rem}}$ as well as a payment vector $\mathbf{p}$ such that the following two properties hold:
\begin{itemize}
\item The outcome $(\alloc, \mathbf{p})$ is envy-free.

\item For each~$i \in N$, we have $v_{k+1} \leq p_i \leq 1$.

(Recall that $v_{k+1} = u_{k+1}(M_{k+1})$.
This property will be useful for the next stage of the algorithm when we allocate the remaining meta-goods $\{M_{k+1}, \dots, M_\ell\}$.)
\end{itemize}

\item The final loop of \Cref{alg:sparse} implements the straightforward second stage.
For each~$j \in \{k+1, \dots, \ell\}$, we give the meta-good~$M_j$ to agent~$j$ and deduct $u_j(M_j)$ amount of subsidy from $p_j$.
It can be seen easily that the updated outcome is envy-free and moreover, each agent receives a subsidy payment of at most~$1$.

Note that $u_j(M_j)$ can be as large as $v_{k+1} = u_{k+1}(M_{k+1})$.
Therefore, $p_j$ must be sufficiently large to ensure that the deduction operation is valid.
This explains the reason why we place a lower bound on the subsidy payment each agent receives in the previous stage of the algorithm.
\end{itemize}

\subsubsection[Allocation of the First Meta-Goods and Residual Chores]{Allocation of $\{M_i\}_{i \in [k]}$ and $Z_{\mathrm{rem}}$}

We now allocate the first $k$ meta-goods and the objective chores $Z_{\mathrm{rem}}=\{c_1,\dots,c_k\}$, where $k<\ell$.
Our goal is to construct an envy-free outcome $(\alloc,\mathbf{p})$ satisfying $v_{k+1}\le p_i\le 1$ for every agent~$i\in N$.
For ease of notation, let us define the following utility threshold:
$$
\lambda \coloneqq v_{k+1}.
$$
We have $0 \leq \lambda < 1$.
Indeed, for any residual chore~$c \in Z_{\mathrm{rem}}$, chore-maximality gives
\[
\lambda + u_{k+1}(c) = u_{k+1}(M_{k+1} \cup \{c\}) < 0,
\]
while $u_{k+1}(c) \geq -1$.
We have $u_i(M_i) \geq \lambda$ for all agents~$i \in [k]$.
Moreover, for any meta-good~$M_j$ with $k+1 \leq j \leq \ell$, every agent~$r \in N$ values~$M_j$ at most~$\lambda$, i.e., $u_r(M_j) \leq \lambda$.

The initialization block and inductive loop of \Cref{alg:sparse} construct the partial outcome $(\alloc, \mathbf{p})$ incrementally.
Starting from the initial set of active agents~$\widetilde{N} = \{1, \dots, k\}$, we allocate $\{M_i\}_{i \in [k]}$ and $Z_{\mathrm{rem}}$ to the active agents in an envy-free manner.
We then iteratively add agents who are not yet active to $\widetilde{N}$, while maintaining the following two invariants at the beginning and end of each iteration:
\begin{enumerate}[label=\textbf{Invariant (I\arabic*):},ref=(I\arabic*),leftmargin=*]
\item\label{invariant:EF}
The allocation $(\alloc, \mathbf{p})$ is envy-free among agents in~$\widetilde{N}$.
Moreover, for every~$i \in \widetilde{N}$,
$$
u_i(A_i) + p_i \geq \lambda, \qquad
u_j(A_i) \leq 0 \text{ for all } j \in N, \quad \text{ and } \quad
\lambda \leq p_i \leq 1.
$$

\item\label{invariant:path}
Every agent in~$\widetilde{N}$ has a directed equality path, in the directed equality graph $\mathcal{H}(\alloc, \mathbf{p})$, to some agent belonging to the following non-empty set
$$
S(\lambda, \mathbf{p}) = \{i \in \widetilde{N} \mid u_i(A_i) + p_i = \lambda\}.
$$
\end{enumerate}

We show that, as long as these invariants hold, it is always possible to add a new agent and restore the invariants if they are broken.
The high-level idea is that we will maintain the lower bound~$\lambda$ on the compensated utility $u_i(A_i) + p_i$ of active agents, while modifying payments and reallocating bundles as new agents are introduced.
Throughout the execution of this process, the set of items $\{M_j\}_{j \in [k]}$ and $Z_{\mathrm{rem}}$ remain allocated; only the set of active agents grows.

\paragraph{Restoring Invariants}
We isolate a technical lemma which helps us restore invariant~\ref{invariant:path} whenever invariant~\ref{invariant:EF} holds and will be used repeatedly in the analysis.
It is worth noting that this can always be enforced by only decreasing payments.
To do so, we need a structural property of envy-free outcomes: agents whose compensated utility exceeds~$\lambda$ can be connected, via tight envy constraints, to agents who attain the bound exactly.

We are now ready to introduce the \emph{directed equality graph}, which captures tight envy constraints between active agents~$\widetilde{N}$.
Note that directed equality graphs are unrelated to the envy graph~$\mathcal{G}_\alloc$ or matching graphs defined in \Cref{def:envy-graph,def:matching-graphs}, respectively.

\begin{definition}[Directed equality graph]
\label{def:equiv-graph}
Given an outcome $(\alloc, \mathbf{p})$ that is envy-free among active agents~$\widetilde{N}$, its \emph{directed equality graph} is the directed graph $\mathcal{H}(\alloc, \mathbf{p})$ on vertex set~$\widetilde{N}$ containing an \emph{equality edge} $i \to j$ whenever agent~$i$ is indifferent between her own bundle and $j$'s bundle under the payments, i.e.,
\[
u_i(A_i) + p_i = u_i(A_j) + p_j.
\]
A \emph{directed equality path} is a directed path in $\mathcal{H}(\alloc, \mathbf{p})$.
\end{definition}

Starting from any envy-free outcome in which each agent's bundle has non-positive utility and compensated utility at least~$\lambda$, \Cref{lem:RestoreI2} below constructs a new envy-free payment vector which never increases any payment and never drops below~$\lambda$, such that every agent has a directed equality path to an agent whose compensated utility is exactly~$\lambda$.
Since later we apply this adjustment to varying subsets of agents, we state the lemma for an arbitrary subset $\widetilde{N} \subseteq N$.

\begin{lemma}
\label{lem:RestoreI2}
Let $0 \leq \lambda <1$.
Consider a set of agents~$\widetilde{N} \subseteq N$ and an outcome $(\alloc, \mathbf{p})$ that is envy-free among agents in~$\widetilde{N}$; moreover, for all~$i \in \widetilde{N}$, we have
\[
u_i(A_i) + p_i \geq \lambda, \qquad  u_i(A_i) \leq 0, \quad \text{ and } \quad  \lambda \leq p_i \leq 1.
\]
Then, an alternative payment vector~$\mathbf{p}'$, with $\lambda \leq p'_i \leq p_i$ for all~$i \in \widetilde{N}$, can be computed in polynomial time such that the following properties hold.
\begin{itemize}
\item Outcome $(\alloc, \mathbf{p}')$ is envy-free among agents in~$\widetilde{N}$.

\item The set $S(\lambda, \mathbf{p}') \coloneqq \{i \in \widetilde{N} \mid u_i(A_i) + p'_i = \lambda\}$ is non-empty.
Moreover, in the directed equality graph $\mathcal{H}(\alloc, \mathbf{p}')$, every agent in~$\widetilde{N}$ has a directed equality path to some agent in~$S(\lambda, \mathbf{p}')$.

\item For all~$i \in \widetilde{N}$, $u_i(A_i) + p'_i \geq \lambda$.
\end{itemize}
\end{lemma}

\begin{proof}
Let $S(\lambda, \mathbf{p}) \coloneqq \{i \in \widetilde{N} \mid u_i(A_i) + p_i = \lambda\}$ denote the set of agents who are getting exactly~$\lambda$.
We first check whether $S(\lambda, \mathbf{p})$ is empty.
If $S(\lambda, \mathbf{p}) = \emptyset$, we uniformly decrease the payments of all agents until the first time some agent satisfies $u_i(A_i) + p_i = \lambda$. Since payments are decreased uniformly among all agents, envy-freeness is preserved.
Thus, without loss of generality, we may assume that $|S(\lambda, \mathbf{p})| \ge 1$.
Note that property $u_i(A_i) + p_i \geq \lambda$ for all $i\in \widetilde{N}$ is maintained by this preprocessing step, since payments are decreased only until the first agent reaches equality.

Let $R \subseteq \widetilde{N}$ denote the set of agents that can reach some agent in $S(\lambda, \mathbf{p})$ via a directed equality path in the equality graph $\mathcal{H}(\alloc, \mathbf{p})$, and let $T \coloneqq \widetilde{N} \setminus R$.
By definition, no agent in~$T$ has a directed equality path to any agent in~$R$.
Moreover, since $(\alloc, \mathbf{p})$ is envy-free, agents in~$T$ do not envy agents in~$R$.

We now uniformly decrease the payments of all agents in~$T$ 
until the first moment at which either an agent in $T$ reaches total value $\lambda$ or an equality edge appears from an agent in $T$ to an agent in $R$.
This operation preserves envy-freeness.
Envy relations among agents in~$T$ are unaffected by a uniform payment decrease, and agents in~$R$ continue not to envy agents in~$T$ since only the payments of agents in~$T$ are reduced.
Furthermore, agents in~$T$ do not envy agents in~$R$: they were initially envy-free towards~$R$, and the payments in~$T$ are decreased only until the first equality edge to an agent in~$R$ appears.
Hence, envy-freeness is preserved throughout the process.

At this point, one of the two stopping events enlarges~$R$.
If an equality edge appears from an agent in~$T$ to an agent in~$R$, then the former agent obtains a directed equality path to~$S(\lambda,\mathbf{p})$.
If instead an agent in~$T$ reaches compensated utility~$\lambda$, then that agent joins~$S(\lambda,\mathbf{p})$ and hence belongs to~$R$ via a path of length zero.
In either case, the size of~$T$ strictly decreases.
Repeating this procedure at most $|\widetilde{N}|$ times yields $T = \varnothing$; that is, every agent has a directed equality path to an agent receiving exactly~$\lambda$, and throughout the process the payments are weakly decreased while envy-freeness is preserved.
Let $\mathbf{p}'$ be the vector of payments constructed, at the end of this process.
By construction, $S(\lambda, \mathbf{p}')$ is non-empty and
every agent has directed equality path in $\mathcal{H}(\alloc, \mathbf{p}')$ to some agent in $S(\lambda, \mathbf{p}')$.

Note that every agent~$i \in \widetilde{N}$ satisfies $u_i(A_i) + p'_i \geq \lambda$.
This holds because the inequality $u_i(A_i) + p_i \geq \lambda$ was satisfied initially.
During the procedure, we decrease an agent's payment only if the agent does not have a directed equality path to any agent receiving a total value of exactly~$\lambda$.
Once such a path appears, the agent's payment no longer decreases.
In particular, if at any point $u_i(A_i) + p_i = \lambda$, then agent~$i$ itself joins~$S(\lambda,\mathbf{p})$ and therefore~$R$.
From that point onward, the agent's payment is never decreased further.
Therefore, the inequality $u_i(A_i) + p'_i \geq \lambda$ is maintained for all agents.

Since each agent's payment is only weakly decreased, we have $p'_i \leq p_i$ for all~$i \in \widetilde{N}$.
Furthermore, because $u_i(A_i) \leq 0$ and $u_i(A_i) + p'_i \geq \lambda$, we have $p'_i \geq \lambda$ for each~$i \in \widetilde{N}$.

\medskip
\noindent\textbf{Polynomial-time Computation.}  
We now explain why the above payment vector $\mathbf{p}'$ can be computed in polynomial time.
The key point is that the set~$R$ grows monotonically.
Once an agent belongs to~$R$, her payment is never decreased again, and the equality path witnessing that she can reach an agent in~$S(\lambda,\mathbf{p})$ is not destroyed by later steps.
In each round with~$T \neq \varnothing$, we decrease only the payments of agents in~$T$ until either an agent from~$T$ reaches compensated utility~$\lambda$ or obtains an equality edge to~$R$, and therefore joins~$R$.
Thus, the set~$R$ strictly increases in every round and can increase at most~$|N|$ times.
Each round only requires constructing the directed equality graph, computing the agents that can reach~$S(\lambda,\mathbf{p})$, and performing the corresponding payment decrease, all of which can be done from the current utilities and payments in polynomial time.
Hence the desired payment vector~$\mathbf{p}'$ can be found in polynomial time.
\end{proof}

With \Cref{lem:RestoreI2} in hand, we proceed with the following observation.

\begin{observation}
\label{obs:restore-I2}
Suppose invariant~\ref{invariant:EF} holds but invariant~\ref{invariant:path} is violated for some allocation
$(\alloc, \mathbf{p})$.
Then, \Cref{lem:RestoreI2} can be applied to restore invariant~\ref{invariant:path}.
Since \Cref{lem:RestoreI2} does not modify the allocation~$\alloc$, the property $u_j(A_i) \leq 0$ for all~$j \in N$ is preserved, and invariant~\ref{invariant:EF} continues to hold.
Consequently, there exists an alternative payment vector that restores~\ref{invariant:path} while preserving~\ref{invariant:EF}, and thus satisfies both invariants simultaneously.
\end{observation}

We are now ready to prove the correctness of the first stage of \Cref{alg:sparse}, which constructs the desired allocation of meta-goods $\{M_1, \dots, M_k\}$ and objective chores $Z_{\mathrm{rem}}$, together with subsidy payments to the agents.

\begin{proposition}
\label{Prop:Partial}
There exists an envy-free allocation $(\alloc, \mathbf{p})$ of the set of items $\{M_i\}_{i \in [k]}$ and $Z_{\mathrm{rem}}$ such that $\lambda \leq p_i \leq 1$ for all~$i \in N$.
Moreover, such an allocation can be found in polynomial time.
\end{proposition}

\begin{proof}
\textbf{Initialization.}
We start with the initial set of active agents $\widetilde{N}_0 = \{1, \dots, k\}$, and allocate $\{M_i\}_{i \in [k]}$ and $Z_{\mathrm{rem}}$  among these agents while satisfying invariants~\ref{invariant:EF} and~\ref{invariant:path}.

First, we allocate the set of chores $Z_{\mathrm{rem}} = \{c_1, \dots, c_k\}$ among the agents in~$\widetilde{N}_0$ using the \IMWPM algorithm (see \Cref{alg:IMWPM}).
The \IMWPM algorithm matches each agent to exactly one chore.
For each agent~$i \in \widetilde{N}_0$, let $c_i^*$ denote the chore assigned the agent.
Since each~$c \in Z_{\mathrm{rem}}$ is an objective chore with utilities $u_i(c) \in [-1, 0)$ for all agents~$i \in \widetilde{N}_0$, \Cref{prop:OC_only} implies that the resulting allocation is envy-freeable among agents in~$\widetilde{N}_0$.
Moreover, there exists a payment vector~$\textbf{p}^0$ satisfying $0 \leq p^0_i \leq 1$ for every~$i \in \widetilde{N}_0$ that establishes envy-freeness:
\[
u_i(c_i^*) + p^0_i \geq u_i(c_j^*) + p^0_j \qquad \text{for all } i,j \in \widetilde{N}_0.
\]

Invariant~\ref{invariant:EF} requires that each agent~$i \in \widetilde{N}$ should have a combined utility of at least~$\lambda$ for their bundle and payment.
We thus increase the agents' payments by a constant simultaneously to achieve this goal.
Let $p^0_{\max} \coloneqq \max_{i \in \widetilde{N}_0} p^0_i$ denote the maximum subsidy in the payment vector~$\textbf{p}^0$.
We shift all payments by the constant $1 - p^0_{\max}$ and define a new payment vector~$\mathbf{p}$ by
\[
p_i \coloneqq p^0_i + (1 - p^0_{\max}) \qquad \text{for all } i \in \widetilde{N}_0.
\]
This uniform shift preserves envy-freeness for agents in~$\widetilde{N}_0$, meaning that
\begin{equation}
\label{eq:sparse:EF}
u_i(c_i^*) + p_i \geq u_i(c_j^*) + p_j \qquad \text{for all}~i, j \in \widetilde{N}_0.
\end{equation}

Moreover, for each~$i \in \widetilde{N}_0$, since $0 \leq p^0_i \leq p^0_{\max} \leq 1$, the new payment satisfies $0 \leq p_i \leq 1$.
Let $r \in \widetilde{N}_0$ be an agent such that $p^0_r = p^0_{\max}$; then $p_r = 1$.
By envy-freeness, for every~$i \in \widetilde{N}_0$, we have
\begin{equation}
\label{eq:sparse:lb-zero}
u_i(c_i^*) + p_i \geq u_i(c_r^*) + p_r \geq -1 + 1 = 0,
\end{equation}
where the last inequality follows from $u_i(c_r^*) \in [-1, 0)$.
We thus obtain an envy-free allocation of the chores $Z_{\mathrm{rem}}$, in which each agent receives exactly one chore, together with subsidies~$\mathbf{p}$ such that
\[
u_i(c_i^*) + p_i \geq 0 \qquad \text{for all } i \in \widetilde{N}_0.
\]

Next, we allocate the subjective meta-goods $\{M_i\}_{i \in [k]}$ to agents in~$\widetilde{N}_0$.
Recall that $\widetilde{N}_0 = [k]$.
For each~$i \in [k]$, we give meta-good~$M_i$ to agent~$i$, and define
\[
A_i \coloneqq \{c_i^*\} \cup M_i.
\]

We now show that the allocation $(\alloc, \mathbf{p})$ satisfies invariant~\ref{invariant:EF}.
By Property~\ref{prop:PThree} of \Cref{lem:algoOneProperties}, each meta-good satisfies \choreMaximality, implying that $u_j(\{c_i^*\} \cup M_i) < 0$ for all~$i, j \in N$, and in particular $u_j(A_i) \leq 0$ for every~$i \in \widetilde{N}_0$ and every~$j \in N$.
Recall that the initial active agents $\widetilde{N}_0 = \{1, \dots, k\}$ are labeled so that $u_i(M_i) \geq \lambda$.
Together with \Cref{eq:sparse:lb-zero}, we have
\[
u_i(A_i) + p_i = u_i(c_i^*) + u_i(M_i) + p_i \geq \lambda.
\]
Since $u_i(\{c_i^*\} \cup M_i) < 0$ for all~$i$, this inequality further implies that for all~$i \in \widetilde{N}_0$, $\lambda \leq p_i \leq 1$.

It remains to verify envy-freeness between agents in~$\widetilde{N}_0$.
By Property~\ref{prop:POne} of \Cref{lem:algoOneProperties}, for all~$i \in \widetilde{N}_0$, the interest sets~$T_i$'s are pairwise disjoint, and thus for any $j \neq i$, we have $u_i(M_j) < 0$.
Hence, envy-freeness follows as for any distinct agents $i, j \in \widetilde{N}_0$,
\begin{align*}
u_i(A_i) + p_i &= u_i(c_i^*) + u_i(M_i) + p_i \\
&\geq u_i(c_i^*) + p_i \tag*{$\because u_i(M_i) \geq 0$} \\
&\geq u_i(c_j^*) + p_j \tag*{$\because$ \Cref{eq:sparse:EF}} \\
&\geq u_i(c_j^*) + u_i(M_j) + p_j \tag*{$\because u_i(M_j) < 0$} \\
&= u_i(A_j) + p_j.
\end{align*}

We conclude that the allocation $(\alloc, \mathbf{p})$ satisfies invariant~\ref{invariant:EF}.

Invariant~\ref{invariant:path} further requires that the set $S(\lambda, \mathbf{p}) = \{i \in \widetilde{N}_0 \mid u_i(A_i) + p_i = \lambda\}$ is non-empty and each agent in~$\widetilde{N}_0$ has a directed equality path in the directed equality graph~$\mathcal{H}(\alloc, \mathbf{p})$ to some agent in~$S(\lambda, \mathbf{p})$.
By \Cref{obs:restore-I2}, we can modify the payments~$\mathbf{p}$ to ensure that both invariants~\ref{invariant:EF} and~\ref{invariant:path} are satisfied simultaneously.

In short, at the end of the initialization block of \Cref{alg:sparse}, we obtain for the initial active set~$\widetilde{N}_0$ an outcome $(\alloc, \mathbf{p})$ satisfying both invariants~\ref{invariant:EF} and~\ref{invariant:path}.

\medskip
\noindent\textbf{Inductive Step.}
We now justify the inductive loop of \Cref{alg:sparse}, which enlarges the active set to include all agents~$N$ while maintaining invariants~\ref{invariant:EF} and~\ref{invariant:path}.
Assume that at the beginning of an iteration, the active set~$\widetilde{N}$ and the allocation $(\widetilde{\alloc}, \mathbf{p})$ satisfy invariants~\ref{invariant:EF} and~\ref{invariant:path}.
We add a new agent~$x \notin \widetilde{N}$ to the active agent set and compute a new allocation for the agents in $\widetilde{N} \cup \{x\}$ such that both invariants~\ref{invariant:EF} and~\ref{invariant:path} are still satisfied.

We first extend the allocation $(\widetilde{\alloc}, \mathbf{p})$ to agents in $\widetilde{N} \cup \{x\}$ by setting $\widetilde{A}_x = \emptyset$ and $p_x = \lambda$.
Note that no agent~$i \in \widetilde{N}$ envies agent~$x$ since by invariant~\ref{invariant:EF}, we have $u_i(\widetilde{A}_i) + p_i \geq \lambda = u_i(\widetilde{A}_x) + p_x$.

If agent~$x$ does not envy any agent in~$\widetilde{N}$, then this extended allocation is envy-free and satisfies invariant~\ref{invariant:EF}.
Invariant~\ref{invariant:path} is also satisfied, since every agent in~$\widetilde{N}$ continues to satisfy \ref{invariant:path} by assumption, and agent~$x$ attains a value of exactly~$\lambda$, and therefore belongs to the set $S(\lambda, \mathbf{p})$ under the new extended allocation.
Thus, in this case both the invariants can be maintained while extending the set of active agents to $\widetilde{N} \cup \{x\}$.

Suppose, on the other hand, agent~$x$ envies someone in~$\widetilde{N}$.
Let $i_1 \in \widetilde{N}$ be the agent that $x$ envies the most, that is, $i_1 = \argmax_{j \in \widetilde{N}} ( u_x(\widetilde{A}_j) + p_j )$.
By invariant~\ref{invariant:path}, agent~$i_1$ has a directed equality path to some agent~$i_s$ in $S(\lambda, \mathbf{p}) = \{i \in \widetilde{N} \mid u_i(\widetilde{A}_i) + p_i = \lambda \}$.
Without loss of generality, denote this directed equality path as
$$
i_1 \rightarrow i_2 \rightarrow \cdots \rightarrow i_s
$$
where, if $i_1 \in S(\lambda, \mathbf{p})$, we may take $s = 1$.

We now define a new allocation $(\alloc, \mathbf{p}')$ over the agent set $\widetilde{N} \cup \{x\}$ as follows:
\[
\begin{cases}
(A_x, p'_x) &= (\widetilde{A}_{i_1}, p_{i_1}), \\
(A_{i_r}, p'_{i_r}) &= (\widetilde{A}_{i_{r+1}}, p_{i_{r+1}}) \qquad \text{for } r \in \{1, 2, \dots, s-1\}, \\
(A_{i_s} , p'_{i_s}) &= (\widetilde{A}_x, p_x) = (\emptyset, \lambda),
\end{cases}
\]
and for all other agents $j \in \widetilde{N} \setminus \{i_1, \dots, i_s\}$, set $(A_j, p'_j) = (\widetilde{A}_j, p_j)$.

We show that the new allocation $(\alloc, \mathbf{p}')$ over $\widetilde{N} \cup \{x\}$ satisfies invariant~\ref{invariant:EF}.
\begin{itemize}
\item For each~$i \in \widetilde{N} \cup \{x\}$, we have $u_j(A_i) \leq 0$ for all~$j \in N$ and $\lambda \leq p'_i \leq 1$.

\item Agent~$x$ is envy-free because, in the envy branch of \Cref{alg:sparse}, she receives the bundle and payment of agent~$i_1$, whom she envied the most.
Moreover, in the allocation $(\widetilde{\alloc}, \mathbf{p})$ agent~$x$ received $(\widetilde{A}_x, p_x) = (\emptyset, \lambda)$ and thus got a utility of exactly~$\lambda$.
It follows easily that $u_x(A_x) + p'_x \geq \lambda$.

\item All agents in~$\widetilde{N}$ receive the same utility under the new allocation $(\alloc, \mathbf{p}')$ as they did under $(\widetilde{\alloc}, \mathbf{p})$.
Each agent~$i_r \in \{i_1, \dots, i_{s-1}\}$ had an equality edge to agent~$i_{r+1}$ under $(\widetilde{\alloc}, \mathbf{p})$; consequently, by receiving the bundle and payment previously assigned to~$i_{r+1}$, agent~$i_r$ attains the same utility as before.
As for agent~$i_s$, since $i_s \in S(\lambda, \mathbf{p})$ under $(\widetilde{\alloc}, \mathbf{p})$, she received utility~$\lambda$.
In the new allocation $(\alloc, \mathbf{p}')$, agent~$i_s$ continues to receive utility~$\lambda$ for her allocation $(\emptyset, \lambda)$, and thus her utility is unchanged.
It follows that for each agent~$i \in \widetilde{N}$, $u_i(A_i) + p'_i = u_i(\widetilde{A}_i) + p_i \geq \lambda$.

Clearly, no agent in~$\widetilde{N}$ envies agent~$i_s$.
Moreover, since we only permute existing bundle-payment pairs, agents in~$\widetilde{N}$ remain envy-free towards any agent in~$\widetilde{N} \cup \{x\}$.
\end{itemize}
We conclude that the allocation $(\alloc,\mathbf{p}')$ is envy-free for all agents in $\widetilde{N} \cup \{x\}$, and furthermore, invariant~\ref{invariant:EF} is maintained.

By \Cref{obs:restore-I2}, if invariant~\ref{invariant:EF} is satisfied, then we can modify the payment vector~$\mathbf{p}'$ to ensure that both invariants~\ref{invariant:EF} and~\ref{invariant:path} are satisfied simultaneously.
Hence, we can extend the set of active agents to $\widetilde{N} \cup \{x\}$ while maintaining both invariants.

In either branch, \Cref{alg:sparse} commits the candidate outcome as the new current state and adds~$x$ to the active set.
Thus, whenever $N \setminus \widetilde{N} \neq \emptyset$, we can add one more agent, update the allocation and payments, and maintain both invariants.
Once all agents~$N$ become active, both invariants hold for~$N$, the proof follows.

\medskip
\noindent\textbf{Polynomial-time Computation.}
We finally explain why the allocation and payments in this proposition can be computed in polynomial time.
The initialization uses the \IMWPM routine on the residual chores and then performs only a constant shift of the resulting payment vector, so this part is polynomial-time.
After the initial outcome is constructed, \Cref{lem:RestoreI2} restores invariant~\ref{invariant:path} in polynomial time.

For the inductive phase, the key progress measure is the active set~$\widetilde{N}$. This set only grows: in each iteration, exactly one new agent~$x$ is added, and no active agent is ever removed. Thus the outer loop runs at most~$n$ times. Within each iteration, we only need to check whether~$x$ envies an active agent, choose a most-envied active agent, find the directed equality path guaranteed by invariant~\ref{invariant:path}, and rotate the bundle--payment pairs along this path. All of these operations can be carried out in polynomial time. After every rotation, \Cref{alg:sparse} invokes \Cref{lem:RestoreI2} on the candidate outcome and enlarged active set, which again takes polynomial time.

Therefore, the algorithm performs only polynomially many polynomial-time operations, and hence computes the desired allocation and payment vector in polynomial time.
\end{proof}

\subsubsection{Completing the Allocation}
Note that, so far, we have only allocated $\{M_i\}_{i \in[k]}$ and $Z_{\mathrm{rem}}$, so we have not yet allocated the subjective meta-goods $\{M_{k+1}, \dots, M_\ell\}$.
Fortunately, most of the heavy lifting is done by \Cref{Prop:Partial}.
The final loop of \Cref{alg:sparse} allocates these remaining items in the straightforward way analyzed in \Cref{thm:CaseII}, thereby completing the allocation without creating envy or increasing the subsidy payments.

\begin{theorem}
\label{thm:CaseII}
Under Case~II, there exists a polynomial-time algorithm that computes an envy-free allocation with subsidies, where each agent receives at most one dollar.
\end{theorem}

\begin{proof}
We start with the envy-free allocation $(\alloc, \mathbf{p})$ of the items $\{M_i\}_{i \in [k]}$ and $Z_{\mathrm{rem}}$ obtained by \Cref{Prop:Partial}.
Recall that, due to the labeling, we have $u_j(M_j) = v_j$ for all~$j \in \{k+1, \dots, \ell\}$.

We construct a complete allocation~$\widetilde{\alloc}$ from~$\alloc$, together with adjusted subsidy payments~$\mathbf{p}'$, as follows:
\begin{equation}
\widetilde{A}_i =
\begin{cases}
A_i \cup M_i & i \in \{k+1, \dots, \ell\} \\
A_i          & i \notin \{k+1, \dots, \ell\}
\end{cases}
\qquad
p'_i =\begin{cases}
p_i - v_i & i \in \{k+1, \dots, \ell\} \\
p_i       & i \notin \{k+1, \dots, \ell\}
\end{cases}
\end{equation}
Note that under~$\widetilde{\alloc}$, every subjective good as well as objective chore is allocated.
We show that the allocation $(\widetilde{\alloc}, \mathbf{p}')$ is envy-free, and satisfies $0 \leq p'_i \leq 1$ for all~$i \in N$.

In both allocations $(\widetilde{\alloc}, \mathbf{p}')$ and $(\alloc, \mathbf{p})$, each agent receives the same value from their own bundle plus payment.
This is immediate for agents in $N \setminus \{k+1, \dots, \ell\}$, whose bundle-payment pairs remain unchanged.
For each agent~$j \in \{k+1, \dots, \ell\}$, since $u_j(M_j) = v_j$, we have
$$
u_j(A_j) + p_j = u_j(\widetilde{A}_j) + p'_j \qquad \text{ for each } j \in \{k+1, \dots, \ell\}.
$$
Thus, every agent attains the same utility for their own bundle payment pair in both allocations.

For a fixed~$j \in \{k+1, \dots, \ell\}$, observe that for all~$i \in N \setminus {j}$, we have $u_i(M_j)-v_j\leq 0$: either $i \notin T_j$, in which we have $u_i(M_j) < 0$, or $i \in T_j$, in which case $u_i(M_j) \leq \max_{r \in T_j} u_r(M_j) = v_j$.
It follows that, no agent~$i \in N$ envies agent~$j \in \{k+1, \dots, \ell\}$ in the allocation $(\widetilde{\alloc}, \mathbf{p}')$, since
\[
u_i(\widetilde{A}_i) + p'_i = u_i(A_i) + p_i \geq u_i(A_j) + p_j \geq u_i(A_j) + p_j + u_i(M_j) - v_j = u_i (\widetilde{A}_j) + p'_j.
\]
Finally, since every agent attains the same value in both allocations, and the allocations and payments of agents in $N \setminus \{k+1, \dots, \ell\}$ remain unchanged, the envy-freeness of $(\alloc, \mathbf{p})$ implies that no agent envies any agent in $N \setminus \{k+1, \dots, \ell\}$ under $(\widetilde{\alloc}, \mathbf{p}')$.
Hence, the allocation $(\widetilde{\alloc}, \mathbf{p}')$ is envy-free.

We now verify $0 \leq p'_i \leq 1$ for all~$i \in N$.
Since $p_i \leq 1$ by \Cref{Prop:Partial}, it follows immediately that $p'_i \leq 1$.
Moreover, by \Cref{Prop:Partial}, we have $p_i \geq \lambda$ for all~$i \in N$.
Because $\lambda \geq v_j$ for each $j \in \{k+1, \dots, \ell\}$, it follows that $p'_i \geq 0$ for all~$i \in N$.

\medskip
\noindent\textbf{Polynomial-time Computation.}
It remains to show that this completion step is polynomial-time.
The preliminary relabeling of meta-goods and agents only requires computing the values~$v_j$ and sorting them.
By \Cref{Prop:Partial}, the partial allocation of $\{M_i\}_{i \in [k]}$ and~$Z_{\mathrm{rem}}$ together with its payment vector can be computed in polynomial time.
After that, the final loop of \Cref{alg:sparse} makes a single pass over the meta-goods $M_{k+1}, \dots, M_\ell$: assign~$M_j$ to agent~$j$ and replace~$p_j$ by~$p_j-v_j$.
Thus, the final completion performs only polynomially many additional operations, and hence the Case~II allocation can be found in polynomial time.
\end{proof}

\subsection[Case III: Abundant Chores]{Case III: Abundant Chores ($|Z_{\mathrm{rem}}| \geq |G|$)}
\label{subsec:abundant-chores}

By the case distinction, upon termination of \Cref{alg:condpairwise}, we have meta-goods \(G = \{M_1, \dots, M_\ell\}\) and a set of objective chores \(Z_{\mathrm{rem}} = \{c_1, \dots, c_k\}\) satisfying $k = |Z_{\mathrm{rem}}|  \geq |G| = \ell$.
Furthermore, by Property~\ref{prop:POne} of \Cref{lem:algoOneProperties}, interest sets $T_j$'s are pairwise disjoint.
Together with $\bigcup_{j \in [\ell]} T_j \subseteq N$, it follows that $\ell \leq n$.
Additionally, by Property~\ref{prop:PThree} of \Cref{lem:algoOneProperties}, combining any meta-good~$M_j \in G$ and any remaining chore~$c \in Z_{\mathrm{rem}}$ results in a \emph{meta-chore}, i.e.,
\[
u_i(M_j \cup c) < 0 \qquad \text{ for all } i \in N.
\]

For any meta-good~$M_j \in G$, if $\max_{i \in T_j} u_i(M_j) = 0$, then $M_j$ can be ignored during the main allocation process.
This is because after allocating all residual objective chores and all other meta-goods in an envy-free manner with subsidies, meta-good~$M_j$ can be assigned arbitrarily to any agent in the interest set~$T_j$.
Since every agent in~$T_j$ values~$M_j$ at~$0$, no envy is created among them. Moreover, for any agent not in~$T_j$, the bundle of the agent receiving~$M_j$ is weakly worse, and hence these agents remain envy-free.

\begin{observation}
\label{obs:non-zero-meta-goods}
For each meta-good~$M_j \in G$, we may assume that
\[
\max_{i \in T_j} u_i(M_j) > 0.
\]
\end{observation}

\paragraph{Pairing Configurations}
Let $\Inj([\ell], [k])$ denote the set of all injective functions $\varphi \colon [\ell] \to [k]$ (i.e., one-to-one assignment of the $\ell$ meta-goods $\{M_j\}_{j \in [\ell]}$ to $\ell$ distinct chores in $Z_{\mathrm{rem}}$).
For each $\varphi \in  \Inj([\ell], [k])$, define the corresponding meta-chore $\widetilde{c}_{\varphi(i)} \coloneqq M_i \cup c_{\varphi(i)}$ for each~$i \in [\ell]$.
Note that each pairing $\varphi \in \Inj([\ell], [k])$  forms a partition of the item set~$M$, denoted as
\[
J^\varphi = \left( \bigcup_{i \in [\ell]} \widetilde{c}_{\varphi(i)} \right) \cup \left( Z_{\mathrm{rem}} \setminus \bigcup_{j \in [\ell]} c_{\varphi(j)} \right).
\]
The first term represents the $\ell$ meta-chores obtained by combining each of the $\ell$ meta-good with a distinct objective chore, while the second term consists of the remaining $k - \ell$ unattached chores from $Z_{\mathrm{rem}}$ that are singletons.
Note that for any $\varphi \in \Inj([\ell], [k])$, the cardinality satisfies $|J^\varphi| = |Z_{\mathrm{rem}}|$.

For later use, note that for each~$i \in [\ell]$, each~$j \in T_i$, and every pairing~$\varphi$, we have
\begin{align}
\label{eq:chorestild}
u_j(\widetilde{c}_{\varphi(i)}) = u_j(M_i \cup c_{\varphi(i)}) \geq u_j(c_{\varphi(i)}) \geq -1,
\end{align}

\paragraph{Roundwise Optimal Pairing}
For each pairing~$\varphi$, let $\IMWPM(H[N, J^\varphi]) = (\mu^1, \dots, \mu^T)$ denote an execution of \IMWPM with respect to the original utility profile~$u$.
(Recall that \IMWPM adds dummy items if $|J^\varphi|$ is not a multiple of~$n$.)
Let $T$-dimensional vector
\[
\val(\IMWPM(H[N, J^\varphi])) = (\val(\mu^1), \dots, \val(\mu^T))
\]
denote the vector of matching values, where in coordinate~$t \in [T]$,
\[
\val(\mu^t) \coloneqq \sum_{i \in N} u_i(\mu^t_i).
\]

We select an injective map~$\varphi^* \in \Inj([\ell], [k])$ and the tie-breaking in \IMWPM so that $\val(\mu^1)$ is maximized over all pairing configurations; subject to the selected matching~$\mu^1$, the value~$\val(\mu^2)$ is maximized over all configurations consistent with~$\mu^1$; and so on.
We call the resulting sequence a \emph{roundwise optimal matching} under~$u$.
Later in \Cref{sec:polytime}, we will describe how to compute the desired injective map~$\varphi^*$ and the sequence of matchings $\IMWPM(H[N, J^{\varphi^*}])$ in polynomial time through flow networks.

\begin{observation}
\label{obs:IMWPT-observation}
When \IMWPM runs for a single round (i.e., $|J^{\varphi^*}| \leq n$), a meta-good paired with some singleton chore may be allocated to some agent outside the interest set of the meta-good.

When \IMWPM runs for at least two rounds (i.e., $|J^{\varphi^*}| \geq n + 1$), every agent is matched to at least one real chore (singleton chore or meta-chore) in the first two rounds.
\end{observation}

We next show that the roundwise optimal matching under~$u$ never assigns an object to an agent who values it below~$-1$.

\begin{lemma}
\label{lem:abundant:interest-assignment}
Let $\alloc^* = (A^*_1, \dots, A^*_n)$ be the allocation produced by the roundwise optimal execution of \IMWPM on~$J^{\varphi^*}$ under~$u$.
Then, in every round~$t$ and for every agent~$i \in N$,
\[
u_i(\mu_i^t) \geq -1.
\]
Equivalently, $u_i(c) \geq -1$ for every non-dummy object~$c \in A_i^*$.
\end{lemma}

\begin{proof}
For each meta-good~$M_j$, fix an agent~$s_j \in T_j$ with $u_{s_j}(M_j)>0$; such an agent exists by \Cref{obs:non-zero-meta-goods}.
The agents~$s_j$ are distinct because the interest sets are pairwise disjoint.

Suppose, for a contradiction, that in some round an agent~$r$ receives an object~$c \cup M_j$ with
\[
u_r(c \cup M_j)<-1.
\]
Every residual singleton chore has value at least~$-1$, so $u_r(M_j)<0$ and hence $r \notin T_j$.
Consider the object received by~$s_j$ in the same round.
If it contains another meta-good~$M_h$, replace~$M_h$ by~$M_j$ and carry~$M_h$ to its chosen agent~$s_h$.
Continue in this way whenever the next chosen agent holds another meta-good.
Because the agents~$s_j$ are distinct and each meta-good has a single holder in the round, the process cannot revisit an earlier assignment unless the next chosen agent is~$r$.
Thus, it either reaches an agent holding a singleton chore or a dummy item, or returns to~$r$ and closes a cycle.

If it reaches a singleton chore, attach the last carried meta-good to that chore and leave~$r$ with the singleton chore~$c$.
If it reaches a dummy item, give the dummy to~$r$ and pair the last carried meta-good with~$c$ for the final chosen agent.
If it returns to~$r$, rotate the carried meta-goods around the resulting cycle.
In all cases, the same chores and the same number of meta-goods are used in this round, so the new pairing is feasible and all other rounds remain unchanged.

Every intermediate chosen agent replaces a meta-good she values negatively by one she values positively.
In the singleton and cycle cases, every affected agent strictly improves.
In the dummy case, agent~$r$ gains more than~$1$, while the final chosen agent loses less than~$1$ because her new combined object has value strictly greater than~$-1$ by \Cref{eq:chorestild}; all intermediate chosen agents strictly improve.
Thus, the value of this round under~$u$ strictly increases, contradicting roundwise optimality.
\end{proof}

\paragraph{Utility Thresholding}
For any pairing~$\varphi$ and every agent~$i \in N$, define the thresholded utility profile~$\widehat{u}$ on~$J^\varphi$ by
\[
\widehat{u}_i(c) \coloneqq \max(-1,u_i(c)) \qquad \text{for every } c \in J^\varphi,
\]
and extend~$\widehat{u}_i$ additively to unions of objects from~$J^\varphi$.
The threshold is used only to obtain the subsidy bound.

\begin{observation}\label{obs:Equival}
For the pairing~$\varphi^*$ and sequence~$(\mu^1,\dots,\mu^T)$ selected above, the same matching sequence is roundwise optimal under~$\widehat{u}$.
Moreover,
\[
\widehat{u}_i(\mu_i^t)=u_i(\mu_i^t) \qquad \text{for every } i \in N \text{ and } t \in [T].
\]
\end{observation}

\begin{proof}
The displayed equality follows immediately from \Cref{lem:abundant:interest-assignment}.
Fix a round after the preceding matchings have been selected, and consider any competing choice~$\nu$ for that round that is consistent with those matchings and can be extended to a complete pairing.
If~$\nu$ assigns some object below~$-1$ under~$u$, apply the same reassignment used in the proof of \Cref{lem:abundant:interest-assignment}.
This reassignment preserves the preceding matchings and uses the same selected chores and the same number of paired meta-goods, so it can still be extended to a complete pairing.
Measured under~$\widehat{u}$, it does not decrease the round value: in the singleton and cycle cases the affected agents weakly improve, while at a dummy endpoint the old bad pair and dummy contribute~$-1+0$ and the new dummy and interested pair contribute strictly more than~$0-1$.
It also removes the assignment below~$-1$.
Repeating if necessary gives a feasible choice~$\nu'$ such that
\[
\sum_{i \in N}\widehat{u}_i(\nu_i)
\leq \sum_{i \in N}\widehat{u}_i(\nu'_i)
= \sum_{i \in N}u_i(\nu'_i)
\leq \sum_{i \in N}u_i(\mu_i^t)
= \sum_{i \in N}\widehat{u}_i(\mu_i^t).
\]
The middle inequality follows from the choice of~$\mu^t$ under~$u$.
Thus, the same round is also optimal under~$\widehat{u}$.
Applying this argument round by round proves the observation.
\end{proof}

Recall that $J^{\varphi^*}$ forms a partition of the item set~$M$.
Under Case III, \Cref{lem:abundant:map-back} below establishes that all items can be allocated in an envy-free manner with each agent receiving a subsidy of at most one dollar.

\begin{lemma}
\label{lem:abundant:map-back}
Let $\alloc^* = (A^*_1, \dots, A^*_n)$ be the allocation produced by the roundwise optimal execution of \IMWPM on~$J^{\varphi^*}$ under~$u$.
Then, there exists a payment vector~$\mathbf{p}^*$ with $0 \le p_i^* \le 1$ for every~$i \in N$ such that $(\alloc^*,\mathbf{p}^*)$ is envy-free.
\end{lemma}

\begin{proof}
By \Cref{obs:Equival}, the same matching sequence that produces~$\alloc^*$ is an execution of \IMWPM under~$\widehat{u}$.
By construction, every object in~$J^{\varphi^*}$ is an objective chore under~$\widehat{u}$ and has value at least~$-1$ for every agent.
It follows from \Cref{prop:OC_only} that there is a payment vector~$\mathbf{p}^*$ with $0 \le p_i^* \le 1$ for all~$i \in N$ such that $(\alloc^*,\mathbf{p}^*)$ is envy-free under~$\widehat{u}$.

We now show that the allocation $(\alloc^*, \mathbf{p}^*)$ is envy-free with respect to the original utility profile~$u$.
For any set $S \subseteq J^{\varphi^*}$ and any agent~$i \in N$, we have $\widehat{u}_i(S) \ge u_i(S)$.
Moreover, by \Cref{lem:abundant:interest-assignment}, for every agent~$i \in N$ and every item~$c \in A_i^*$, we have $u_i(c) \ge -1$.
It follows that $\widehat{u}_i(c) = u_i(c)$ for all~$c \in A_i^*$ and all~$i \in N$.
Hence, $\widehat{u}_i(A^*_i) = u_i(A^*_i)$ for every agent.
Envy-freeness with respect to utility profile~$u$ follows since, for any two agents~$i, j \in N$, we have
\[
u_i(A_i^*) + p_i^* = \widehat{u}_i(A_i^*) + p_i^* \geq \widehat{u}_i(A_j^*) + p_j^* \geq u_i(A_j^*) + p_j^*,
\]
where the first inequality uses envy-freeness under~$\widehat{u}$.
\end{proof}

\section{Polynomial-Time Computation}
\label{sec:polytime}

In this section, we explain how to implement the construction in the proof of \Cref{thm:main} in polynomial time.
The bundling procedure in \Cref{alg:condpairwise} is polynomial-time.
Case~I is handled by running \IMWM on the resulting subjective-goods instance and then computing the usual longest-path payments in the envy graph.
Case~II is polynomial-time by the local arguments in \Cref{subsec:sparse-chores}, in particular \Cref{Prop:Partial} and \Cref{thm:CaseII}.
The remainder of this section is to explain how to implement the abundant-chores construction from \Cref{subsec:abundant-chores}.

In \Cref{subsec:abundant-chores}, the existence proof selects an injective map~$\varphi^* \in \Inj([\ell], [k])$.
The resulting matching $\IMWPM(H[N, J^{\varphi^*}]) = (\mu^1, \dots, \mu^T)$ has a value vector under~$u$ in which $\val(\mu^1)$ is maximized, and next, subject to the selected matching~$\mu^1$, $\val(\mu^2)$ is maximized, and so on.
As in \Cref{obs:non-zero-meta-goods}, meta-goods \(M_j\) with \(\max_{i\in T_j}u_i(M_j)=0\) can be set aside and assigned at the end to an arbitrary agent in \(T_j\).
This does not create envy and is clearly polynomial-time.
Thus, as in \Cref{subsec:abundant-chores}, we assume throughout the algorithm below that every remaining meta-good \(M_j\) satisfies
\[
\max_{i \in T_j} u_i(M_j) > 0.
\]

At a high level, our algorithm constructs a sequence of flow networks consisting of the agents~$N$, (subsets of) the meta-goods $\{M_1, \dots, M_\ell\}$ and (subsets of) the residual objective chores~$Z_{\mathrm{rem}}$ so that a flow of maximum profit on each such network represents some maximum-weight matching and thus finds a roundwise optimal matching.

We first introduce \emph{residual subproblem}, which consists of remaining real chores \(Z'\subseteq Z_{\mathrm{rem}}\) and remaining meta-goods \(G'\subseteq G\), with \(Z'\neq \emptyset\) and \(|Z'|\ge |G'|\), to be allocated among agents~$N$.
The structural properties from \Cref{lem:algoOneProperties} are inherited by this residual subproblem: the interest sets remain pairwise disjoint, and every pair \(c\cup M_j\), with \(c\in Z'\) and \(M_j\in G'\), is a meta-chore.
Let
\begin{equation}
\label{eq:q}
q \coloneqq
\begin{cases}
|Z'| \bmod n, & \text{if } |Z'|\not\equiv 0 \pmod n,\\
n, & \text{otherwise,}
\end{cases}
\qquad
b \coloneqq |Z'|-|G'|.
\end{equation}
The integer~$q$ is the number of real chores that must be selected when applying one round of \IMWPM to this residual subproblem.
The integer~$b$ is the total number of remaining real chores that need not be paired with meta-goods.
A \emph{feasible next round} selects~$q$ chores, assigns them to distinct agents, and pairs some of them with distinct meta-goods.
If $p$ meta-goods are paired, we call the round \emph{extendable} when
\[
|Z'| - q \ge |G'| - p,
\]
or equivalently when $q - p \leq |Z'| - |G'| = b$.
This condition means that there are still enough real chores to pair all remaining meta-goods after the round.

\begin{figure}[t]
\centering
\begin{tikzpicture}[
>=latex,
scale=1.14,
transform shape,
vertex/.style={circle, draw, line width=0.45pt, minimum size=8.5mm, inner sep=0pt, font=\small},
arc/.style={->, line width=0.6pt},
fam/.style={->, line width=0.42pt, densely dashed, ptgray},
pairarc/.style={->, line width=0.5pt, densely dashed, cyan},
bypass/.style={->, line width=0.55pt},
title/.style={font=\footnotesize\bfseries, align=center}
]
\node[vertex] (s) at (0,0) {$s$};
\node[vertex] (sp) at (1.35,0) {$s'$};

\node[vertex] (c1) at (3.05,1.25) {$c_1$};
\node[vertex] (c2) at (3.05,0) {$c_2$};
\node[vertex] (c3) at (3.05,-1.25) {$c_{|Z'|}$};

\node[vertex] (r1m) at (5.35,1.25) {$r_1^-$};
\node[vertex] (r2m) at (5.35,0) {$r_2^-$};
\node[vertex] (r3m) at (5.35,-1.25) {$r_n^-$};

\node[vertex] (r1p) at (6.85,1.25) {$r_1^+$};
\node[vertex] (r2p) at (6.85,0) {$r_2^+$};
\node[vertex] (r3p) at (6.85,-1.25) {$r_n^+$};

\node[vertex] (m1) at (9.15,1.0) {$m_1$};
\node[vertex] (m2) at (9.15,-0.4) {$m_{|G'|}$};
\node[vertex] (tp) at (9.15,-1.55) {$t'$};
\node[vertex] (t) at (11.05,0.15) {$t$};

\node[title] at (0,2.45) {source};
\node[title] at (1.35,2.45) {round\\selector};
\node[title] at (3.05,2.45) {real chores};
\node[title] at (5.35,2.45) {agent\\input};
\node[title] at (6.85,2.45) {agent\\output};
\node[title] at (9.15,2.45) {meta-goods};
\node[title] at (11.05,2.45) {sink};

\node[font=\scriptsize, black!55] at (3.05,-0.63) {\(\vdots\)};
\node[font=\scriptsize, black!55] at (5.35,-0.63) {\(\vdots\)};
\node[font=\scriptsize, black!55] at (6.85,-0.63) {\(\vdots\)};
\node[font=\scriptsize, black!55] at (9.15,0.3) {\(\vdots\)};

\draw[arc] (s) -- (sp);
\draw[arc] (sp) -- (c1);
\draw[arc] (sp) -- (c2);
\draw[arc] (sp) -- (c3);

\foreach \c in {c1,c2,c3} {
    \foreach \r in {r1m,r2m,r3m} {
        \draw[fam] (\c) -- (\r);
    }
}

\draw[arc] (r1m) -- (r1p);
\draw[arc] (r2m) -- (r2p);
\draw[arc] (r3m) -- (r3p);

\foreach \r in {r1p,r2p,r3p} {
    \foreach \m in {m1,m2} {
        \draw[pairarc] (\r) -- (\m);
    }
}

\draw[arc] (m1) -- (t);
\draw[arc] (m2) -- (t);

\draw[bypass] (r1p) to[out=-32,in=165] (tp);
\draw[bypass] (r2p) to[out=-30,in=180] (tp);
\draw[bypass] (r3p) to[out=-18,in=200] (tp);
\draw[arc] (tp) to[out=0,in=-100] (t);
\end{tikzpicture}

\vspace{1em}
{\footnotesize
\begin{tabular}{@{}lll@{\hspace{2cm}}lll@{}}
\toprule
Arc family & Capacity & Profit & Arc family & Capacity & Profit \\
\midrule
$s \to s'$ & $q$ & $0$ &
$r^+ \to m_j$ & $1$ & $u_r(M_j)$ \\
$s' \to c$ & $1$ & $0$ &
$r^+ \to t'$ & $1$ & $0$ \\
$c \to r^-$ & $1$ & $u_r(c)$ &
$m_j \to t$ & $1$ & $0$ \\
$r^- \to r^+$ & $1$ & $0$ &
$t' \to t$ & $b$ & $0$ \\
\bottomrule
\end{tabular}
}
\caption{The max-profit flow network for the residual subproblem parametrized with $Z'$ and $G'$.
The light dashed arcs are the $c \to r^-$ complete bipartite family, and the blue dashed arcs are the $r^+ \to m_j$ complete bipartite family.
Flow through a meta-good node pairs the selected chore with that meta-good; flow through $t'$ leaves the selected chore unpaired.}
\label{fig:round-flow-network}
\end{figure}
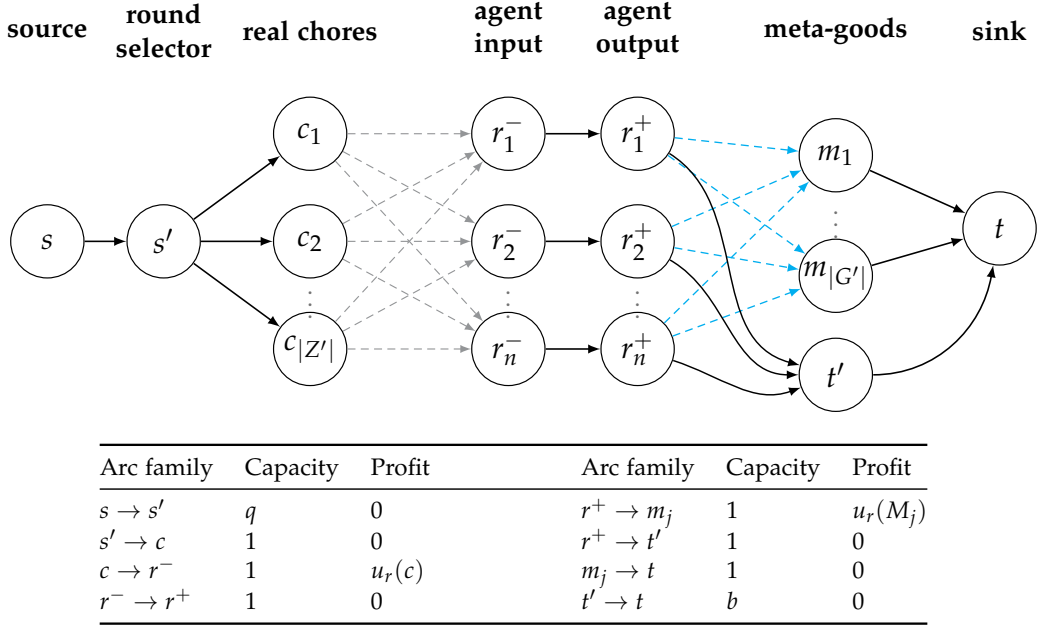

Next, we construct a flow network to capture the residual subproblem; see \Cref{fig:round-flow-network} for a pictorial illustration.
The network has source~$s$, auxiliary source~$s'$, sink~$t$, auxiliary sink~$t'$, one node for each chore~$c \in Z'$, two nodes~$r^-$ and~$r^+$ for each agent~$r \in N$, and one node~$m_j$ for each meta-good~$M_j \in G'$.
Furthermore, the network contains the following arcs:
\begin{itemize}
\item $s \to s'$, with capacity~$q$ and profit~$0$;
\item $s' \to c$, with capacity~$1$ and profit~$0$, for each~$c \in Z'$;
\item $c \to r^-$, with capacity~$1$ and profit~$u_r(c)$, for each~$c \in Z'$ and~$r \in N$;
\item $r^- \to r^+$, with capacity~$1$ and profit~$0$, for each~$r \in N$;
\item $r^+ \to m_j$, with capacity~$1$ and profit~$u_r(M_j)$, for each~$r \in N$ and~$M_j \in G'$;
\item $r^+ \to t'$, with capacity~$1$ and profit~$0$, for each~$r \in N$;
\item $m_j \to t$, with capacity~$1$ and profit~$0$, for each~$M_j \in G'$;
\item $t' \to t$, with capacity~$b$ and profit~$0$.
\end{itemize}

Consider only integral flows on the network.
If a unit of flow uses the path
\[
s \to s' \to c \to r^- \to r^+ \to m_j \to t,
\]
it represents that agent~$r$ receives the paired object $c \cup M_j$.
Put differently, chore~$c$ and meta-good~$M_j$ are paired together and then matched to agent~$r$.
On the other hand, if a unit of flow uses the path
\[
s \to s' \to c \to r^- \to r^+ \to t' \to t,
\]
then agent~$r$ receives (or is matched to) the singleton chore~$c$.
Thus, node~$t'$ is a counter for selected chores that are left unpaired.

We send exactly~$q$ units of flow from~$s$ to~$t$.
It selects~$q$ chores from~$Z'$, assigns them to distinct agents, and pairs some of the selected chores with distinct meta-goods from~$G'$.
Selected chores not paired with a meta-good remain singleton chores.
Agents not assigned a real chore receive dummy items.
Each unit of flow through node~$t'$ is one selected singleton chore left unpaired in the current residual subproblem, and the capacity $b = |Z'| - |G'|$ on the edge~$t' \to t$ ensures that the subsequent residual subproblem remains valid.

\begin{lemma}
\label{lem:round-flow}
For every residual subproblem~$(Z',G')$, integral feasible flows of value~$q$ in the network above are in value-preserving correspondence with extendable feasible next rounds.
Consequently, a maximum-profit integral flow maximizes the value under~$u$ of the next round among all extendable feasible next rounds.
\end{lemma}

\begin{proof}
Given an extendable feasible next round, route each selected chore~$c$ through the agent~$r$ who receives it.
If~$c$ is paired with~$M_j$, continue through~$m_j$; otherwise continue through~$t'$.
The capacities enforce that each selected chore is used once, each agent receives at most one real object, each meta-good is paired at most once, and at most $b=|Z'|-|G'|$ selected chores are left unpaired.
The last condition is exactly the extendability condition $q-p\le b$, where~$p$ is the number of paired meta-goods.

Conversely, since the network is acyclic and all capacities are integral, any integral flow of value~$q$ decomposes into~$q$ paths of one of the two displayed forms.
Reading off these paths gives a feasible next round.
If the flow pairs~$p$ meta-goods, then $q-p$ units use the arc~$t'\to t$, so
\[
q-p\le b=|Z'|-|G'|.
\]
Equivalently, $|Z'|-q\ge |G'|-p$, and hence the residual instance remains extendable.

Finally, the profit of a path assigning~$c$ to~$r$ and then passing through~$t'$ is $u_r(c)$, while the profit of a path assigning $c\cup M_j$ to~$r$ is $u_r(c)+u_r(M_j)$.
Thus, the flow profit is exactly the value under~$u$ of the corresponding next round.
\end{proof}

\begin{proposition}
\label{prop:caseIII-polytime}
Under Case~III, one can compute in polynomial time an allocation and payments satisfying the conclusions of the abundant-chores construction.
\end{proposition}

\begin{proof}
Initialize $Z^1 \coloneqq Z_{\mathrm{rem}}$ and $G^1 \coloneqq G$, after setting aside the zero-valued meta-goods as in \Cref{obs:non-zero-meta-goods}.
At round~$h$, if $Z^h=\emptyset$, stop.
Otherwise, let $q_h$ and $b_h$ be the values of~$q$ and~$b$ from \Cref{eq:q} for the residual subproblem~$(Z^h,G^h)$.
Compute a maximum-profit integral flow of value~$q_h$, allocate the corresponding singleton chores and paired chore--meta-good objects, and remove the selected chores and paired meta-goods to obtain~$(Z^{h+1},G^{h+1})$.

The invariant $|Z^h|\ge |G^h|$ is preserved.
Indeed, if the round pairs~$p_h$ meta-goods, then the capacity of~$t'\to t$ gives
\[
q_h-p_h\le |Z^h|-|G^h|,
\]
which is equivalent to
\[
|Z^h|-q_h\ge |G^h|-p_h.
\]
Thus the next residual subproblem is again abundant.
Each round removes at least one chore, so the procedure terminates; when $Z^h=\emptyset$, the invariant implies $G^h=\emptyset$ as well.
Consequently, every remaining positive-valued meta-good is paired with a distinct real chore, and the completed procedure determines a pairing~$\varphi^*$ and its object collection~$J^{\varphi^*}$.

We next show that the selected rounds form a roundwise optimal execution of \IMWPM under~$u$.
Fix a round~$h$ and consider the final collection of singleton chores and paired objects determined by continuing the procedure to termination.
At the beginning of round~$h$, its remaining real objects correspond exactly to the residual subproblem~$(Z^h,G^h)$.
After the dummy objects required by \IMWPM are added, there are $n-q_h$ dummy objects available in this round.
Every real object is an objective chore under~$u$, whereas every dummy object has value~$0$, so a maximum-weight perfect matching can be chosen to use all available dummy objects.
Its non-dummy part therefore selects exactly~$q_h$ real objects.
If it selects~$p_h'$ paired objects, then it selects $q_h-p_h'$ singleton chores; because the completed residual collection contains only $|Z^h|-|G^h|$ singleton chores, this is an extendable feasible next round.

Conversely, the round selected by the flow, together with dummy objects assigned to agents receiving no real object, is a feasible perfect matching in this same completed collection.
By \Cref{lem:round-flow}, it maximizes the value under~$u$ among all extendable feasible next rounds and hence is a maximum-weight perfect matching for the completed collection.
Maximizing over extendable rounds is exactly maximizing over pairing configurations consistent with the preceding rounds.
Applying this argument round by round therefore shows that the flow procedure produces a roundwise optimal execution of \IMWPM under~$u$.

It remains to compute the payments.
By \Cref{lem:abundant:interest-assignment}, every non-dummy object assigned by this execution has original value at least~$-1$ for its recipient, and \Cref{obs:Equival} shows that the same matching sequence is an execution of \IMWPM under~$\widehat{u}$.
The longest-path payments under~$\widehat{u}$ can therefore be computed in polynomial time, and \Cref{lem:abundant:map-back} shows that they give an envy-free outcome under~$u$ with each payment at most~$1$.
Finally, assign the zero-valued meta-goods set aside at the beginning as in \Cref{obs:non-zero-meta-goods}.

There are at most $\lceil |Z_{\mathrm{rem}}|/n\rceil$ flow rounds, and each network has polynomially many nodes and arcs.
Negating the rational profits gives a minimum-cost flow problem with integral capacities and rational costs of polynomial bit length.
Standard minimum-cost flow algorithms return an integral optimum in time polynomial in the input bit length~\citep[e.g.,][]{AhujaMaOr93}.
The final longest-path computation is also polynomial-time.
\end{proof}

Combining \Cref{prop:caseIII-polytime} with the polynomial-time arguments for Cases~I and~II gives a polynomial-time implementation of the full construction in \Cref{thm:main}.

\section{Discussion}

In this paper, we have studied the fair allocation of indivisible items that may be goods for some agents and chores for others under the subsidy model, and proved that one dollar per agent always suffices to guarantee an envy-free outcome for additive utilities normalized to $[-1, 1]$ (Theorem~\ref{thm:main}).
Since the bound is tight already for goods-only instances~\citep{BrustleDiNa20}, it is worst-case optimal in the mixed setting as well.
To obtain the guarantee, we connect the envy-graph characterization of \EFability and longest-path subsidies~\citep{HalpernSh19} with matching-based allocations (\IMWM\ / \IMWPM) and a bundling step that reduces a mixed instance to meta-goods with desirable properties, together with residual objective chores.

A natural direction for future work is to move beyond additive utilities.
Our techniques rely heavily on additivity both to define meta-goods and to obtain telescoping envy-path bounds from iterated matchings.
This approach does not directly extend to richer valuation classes.
This raises the following questions:
for mixed goods and chores with bounded marginal values, does a constant per-agent subsidy bound still hold for submodular or XOS valuations?
Note that this question is open even in the goods-only setting.

\section*{Acknowledgements}

This work was partially supported
by the NSF-CSIRO grant on ``Fair Sequential Collective Decision Making'',
by the ARC Laureate Project FL200100204 on ``Trustworthy AI'', and
by JST ERATO Grant Number JPMJER2301.

\bibliographystyle{plainnat}
\bibliography{bibliography}

@article{AlkanDeGa91,
	author = {Alkan, Ahmet and Demange, Gabrielle and Gale, David},
	title = {Fair Allocation of Indivisible Goods and Criteria of Justice},
	journal = {Econometrica},
	volume = {59},
	number = {4},
	year = {1991},
	pages = {1023--1039},
}

@article{AmanatidisAzBi23,
	author = {Amanatidis, Georgios and Aziz, Haris and Birmpas, Georgios and Filos-Ratsikas, Aris and Li, Bo and Moulin, Herv\'{e} and Voudouris, Alexandros A. and Wu, Xiaowei},
	title = {Fair Division of Indivisible Goods: Recent Progress and Open Questions},
	journal = {Artificial Intelligence},
	year = {2023},
	volume = {322},
	pages = {103965},
}

@article{Aragones95,
	author = {Aragones, Enriqueta},
	title = {A Derivation of the Money {R}awlsian Solution},
	journal = {Social Choice and Welfare},
	volume = {12},
	number = {3},
	year = {1995},
	pages = {267--276},
}

@article{AzizCaIg22,
	author = {Aziz, Haris and Caragiannis, Ioannis and Igarashi, Ayumi and Walsh, Toby},
	title = {Fair Allocation of Indivisible Goods and Chores},
	journal = {Autonomous Agents and Multi-Agent Systems},
	volume = {36},
	number = {1},
	year = {2022},
	pages = {3:1--3:21},
}

@inproceedings{AzizLuMa26,
	author = {Aziz, Haris and Lu, Xinhang and Mackenzie, Simon and Suzuki, Mashbat},
	title = {Fair Division with Indivisible Goods, Chores, and Cake},
	booktitle = {Proceedings of the 27th ACM Conference on Economics and Computation (EC)},
	year = {2026},
	note = {Forthcoming},
}

@inproceedings{BarmanHVSe25,
	author = {Barman, Siddharth and HV, Vishwa Prakash and Sethia, Aditi and Suzuki, Mashbat},
	title = {Fair and Efficient Allocation of Indivisible Mixed Manna},
	booktitle = {Proceedings of the 21st Conference on Web and Internet Economics (WINE)},
	year = {2025},
	pages = {467--483},
}

@inproceedings{BarmanKrNa22,
	author = {Barman, Siddharth and Krishna, Anand and Narahari, Yadati and Sadhukhan, Soumyarup},
	title = {Achieving Envy-Freeness with Limited Subsidies under Dichotomous Valuations},
	booktitle = {Proceedings of the 31st International Joint Conference on Artificial Intelligence (IJCAI)},
	year = {2022},
	pages = {60--66},
}

@article{BarmanVe26,
	author = {Barman, Siddharth and Verma, Paritosh},
	title = {Introspectively Envy-Free and Efficient Allocation of Indivisible Mixed Manna},
	journal = {CoRR},
	volume = {abs/2509.18673},
	year = {2026},
}

@article{BogomolnaiaMoSa17,
	author = {Bogomolnaia, Anna and Moulin, Herv\'{e} and Sandomirskiy, Fedor and Yanovskaya, Elena},
	title = {Competitive Division of a Mixed Manna},
	journal = {Econometrica},
	volume = {85},
	number = {6},
	year = {2017},
	pages = {1847--1871},
}

@inproceedings{BrustleDiNa20,
	author = {Brustle, Johannes and Dippel, Jack and Narayan, Vishnu V. and Suzuki, Mashbat and Vetta, Adrian},
	title = {One Dollar Each Eliminates Envy},
	booktitle = {Proceedings of the 21st ACM Conference on Economics and Computation (EC)},
	year = {2020},
	pages = {23--39},
}

@article{ChaudhuryGaMc23,
	author = {Chaudhury, Bhaskar Ray and Garg, Jugal and McGlaughlin, Peter and Mehta, Ruta},
	title = {A Complementary Pivot Algorithm for Competitive Allocation of a Mixed Manna},
	journal = {Mathematics of Operations Research},
	volume = {48},
	number = {3},
	year = {2023},
	pages = {1630--1656},
}

@article{ElmalemAzGo25,
    author = {Elmalem, Noga Klein and Aziz, Haris and Gonen, Rica and Huang, Xin and Kimura, Kei and Saha, Indrajit and Segal-Halevi, Erel and Sun, Zhaohong and Suzuki, Mashbat and Yokoo, Makoto},
    title = {Whoever Said Money Won't Solve All Your Problems? {W}eighted Envy-Free Allocation with Subsidy},
    journal = {CoRR},
    volume = {abs/2502.09006},
    year = {2025},
}

@article{Foley67,
	author = {Foley, Duncan Karl},
	title = {Resource Allocation and the Public Sector},
	journal = {Yale Economics Essays},
	volume = {7},
	number = {1},
	year = {1967},
	pages = {45--98},
}

@article{GokoIgKa24,
	author = {Goko, Hiromichi and Igarashi, Ayumi and Kawase, Yasushi and Makino, Kazuhisa and Sumita, Hanna and Tamura, Akihisa and Yokoi, Yu and Yokoo, Makoto},
	title = {A Fair and Truthful Mechanism with Limited Subsidy},
	journal = {Games and Economic Behavior},
	volume = {144},
	year = {2024},
	pages = {49--70},
}

@article{GuoLiDe23,
	author = {Guo, Hao and Li, Weidong and Deng, Bin},
	title = {A Survey on Fair Allocation of Chores},
	journal = {Mathematics},
	volume = {11},
	number = {16},
	year = {2023},
	pages = {3616},
}

@inproceedings{HalpernSh19,
	author = {Halpern, Daniel and Shah, Nisarg},
	title = {Fair Division with Subsidy},
	booktitle = {Proceedings of the 12th International Symposium on Algorithmic Game Theory (SAGT)},
	year = {2019},
	pages = {374--389},
}

@inproceedings{HosseiniSiVa23,
	author = {Hosseini, Hadi and Sikdar, Sujoy and Vaish, Rohit and Xia, Lirong},
	title = {Fairly Dividing Mixtures of Goods and Chores under Lexicographic Preferences},
	booktitle = {Proceedings of the 22nd International Conference on Autonomous Agents and Multiagent Systems (AAMAS)},
	year = {2023},
	pages = {152--160},
}

@article{Klijn00,
	author = {Klijn, Flip},
	title = {An Algorithm for Envy-free Allocations in an Economy with Indivisible Objects and Money},
	journal = {Social Choice and Welfare},
	volume = {17},
	number = {2},
	year = {2000},
	pages = {201--215},
}

@inproceedings{KulkarniMeTa21,
	author = {Kulkarni, Rucha and Mehta, Ruta and Taki, Setareh},
	title = {Indivisible Mixed Manna: On the Computability of {MMS} + {PO} Allocations},
	booktitle = {Proceedings of the 22nd ACM Conference on Economics and Computation (EC)},
	year = {2021},
	pages = {683--684},
}

@article{LiSuSu25,
	author = {Li, Bo and Sun, Ankang and Suzuki, Mashbat and Xing, Shiji},
	title = {On the Subsidy of Envy-Free Orientations in Graphs},
	journal = {CoRR},
	volume = {abs/2502.13671},
	year = {2025},
}

@article{LiuLuSu24,
	author = {Liu, Shengxin and Lu, Xinhang and Suzuki, Mashbat and Walsh, Toby},
	title = {Mixed Fair Division: {A} Survey},
	journal = {Journal of Artificial Intelligence Research},
	volume = {80},
	year = {2024},
	pages = {1373--1406},
}

@incollection{Maskin87,
	author = {Maskin, Eric S.},
	title = {On the Fair Allocation of Indivisible Goods},
	booktitle = {Proceedings of the Arrow and the Foundations of the Theory of Economic Policy},
	editor = {Feiwel, George R.},
	publisher = {Palgrave Macmillan UK},
	year = {1987},
	pages = {341--349},
}

@article{MeertensPoRe02,
	author = {Meertens, Marc and Potters, Jos and Reijnierse, Hans},
	title = {Envy-free and {P}areto Efficient Allocations in Economies with Indivisible Goods and Money},
	journal = {Mathematical Social Sciences},
	volume = {44},
	number = {3},
	year = {2002},
	pages = {223--233},
}

@inproceedings{NarayanSuVe21,
	author = {Narayan, Vishnu V. and Suzuki, Mashbat and Vetta, Adrian},
	title = {Two Birds with One Stone: Fairness and Welfare via Transfers},
	booktitle = {Proceedings of the 14th International Symposium on Algorithmic Game Theory (SAGT)},
	year = {2021},
	pages = {376--390},
}

@inproceedings{NguyenRo23,
	author = {Nguyen, Trung Thanh and Rothe, J\"{o}rg},
	title = {Complexity Results and Exact Algorithms for Fair Division of Indivisible Items: {A} Survey},
	booktitle = {Proceedings of the 32nd International Joint Conference on Artificial Intelligence (IJCAI)},
	year = {2023},
	pages = {6732--6740},
}

@article{Suksompong21,
	author = {Suksompong, Warut},
	title = {Constraints in Fair Division},
	journal = {ACM SIGecom Exchanges},
	volume = {19},
	number = {2},
	year = {2021},
	pages = {46--61},
}

@article{Suksompong25,
	author = {Suksompong, Warut},
	title = {Weighted Fair Division of Indivisible Items: {A} Review},
	journal = {Information Processing Letters},
	volume = {187},
	year = {2025},
	pages = {106519},
}

@article{Svensson83,
	author = {Svensson, Lars-Gunnar},
	title = {Large Indivisibles: An Analysis with Respect to Price Equilibrium and Fairness},
	journal = {Econometrica},
	volume = {51},
	number = {4},
	year = {1983},
	pages = {939--954},
}

@inproceedings{WuXuZh25,
	author = {Wu, Xiaowei and Xue, Quan and Zhou, Shengwei},
	title = {A Little Subsidy Ensures MMS Allocation for Three Agents},
	booktitle = {Proceedings of the 34th International Joint Conference on Artificial Intelligence (IJCAI)},
	year = {2025},
	pages = {4073--4081},
}

@inproceedings{WuXuZh25-revisiting,
	author = {Wu, Xiaowei and Xue, Quan and Zhou, Shengwei},
	title = {Revisiting Proportional Allocation with Subsidy: Simplification and Improvements},
	booktitle = {Proceedings of the 34th International Joint Conference on Artificial Intelligence (IJCAI)},
	year = {2025},
	pages = {4082--4090},
}

@inproceedings{WuZh24,
	author = {Wu, Xiaowei and Zhou, Shengwei},
	title = {Tree Splitting Based Rounding Scheme for Weighted Proportional Allocations with Subsidy},
	booktitle = {Proceedings of the 20th Conference on Web and Internet Economics (WINE)},
	year = {2024},
	pages = {295--313},
}

@inproceedings{WuZhZh23,
	author = {Wu, Xiaowei and Zhang, Cong and Zhou, Shengwei},
	title = {One Quarter Each (on Average) Ensures Proportionality},
	booktitle = {Proceedings of the 19th Conference on Web and Internet Economics (WINE)},
	year = {2023},
	pages = {582--599},
}

@article{KawaseMaSuTaYo25,
  author  = {Kawase, Yasushi and Makino, Kazuhisa and Sumita, Hanna and Tamura, Akihisa and Yokoo, Makoto},
  title   = {Towards Optimal Subsidy Bounds for Envy-Freeable Allocations},
  journal = {Artificial Intelligence},
  volume  = {348},
  pages   = {104406},
  year    = {2025},
}

@article{Su99,
  author  = {Su, Francis Edward},
  title   = {Rental Harmony: {Sperner's} Lemma in Fair Division},
  journal = {The American Mathematical Monthly},
  volume  = {106},
  number  = {10},
  pages   = {930--942},
  year    = {1999},
}

@inproceedings{LiptonMaMoSa04,
  author    = {Lipton, Richard J. and Markakis, Evangelos and Mossel, Elchanan and Saberi, Amin},
  title     = {On Approximately Fair Allocations of Indivisible Goods},
  booktitle = {Proceedings of the 5th ACM Conference on Electronic Commerce (EC)},
  year      = {2004},
  pages     = {125--131},
}

@article{Budish11,
  author  = {Budish, Eric},
  title   = {The Combinatorial Assignment Problem: Approximate Competitive Equilibrium from Equal Incomes},
  journal = {Journal of Political Economy},
  volume  = {119},
  number  = {6},
  pages   = {1061--1103},
  year    = {2011},
}

@article{CaragiannisKuMoPrShWa19,
  author    = {Caragiannis, Ioannis and Kurokawa, David and Moulin, Herv{\'e} and Procaccia, Ariel D. and Shah, Nisarg and Wang, Junxing},
  title     = {The Unreasonable Fairness of Maximum Nash Welfare},
  journal   = {ACM Transactions on Economics and Computation},
  volume    = {7},
  number    = {3},
  articleno = {12},
  numpages  = {32},
  year      = {2019},
}

@book{AhujaMaOr93,
	author = {Ahuja, Ravindra K. and Magnanti, Thomas L. and Orlin, James B.},
	title = {Network Flows: Theory, Algorithms, and Applications},
	publisher = {Prentice Hall},
	address = {Englewood Cliffs, NJ},
	year = {1993},
	isbn = {978-0-13-617549-0},
	url = {https://www.pearson.com/en-us/subject-catalog/p/network-flows-theory-algorithms-and-applications/P200000003456},
}

\appendix

\section{Omitted Proofs}

\subsection{Proof of Proposition~\ref{prop:OC_only}}
\label{app:prop32}

Let $\mathcal{A}=(A_1,\dots,A_n)$ be the allocation returned by \IMWPM. For each round~$t$, let $\mu^t=\{\mu^t_1,\dots,\mu^t_n\}$ denote the items allocated in that round, and let $J^t$ be the set of items available at its beginning.

Fix a path $P$ in the envy graph $\mathcal{G}_\mathcal{A}$; without loss of generality, let this directed path be $1\rightarrow 2\rightarrow \cdots \rightarrow k$. Then
\begin{align*}
w_{\mathcal{A}}(P) = \sum_{i=1 }^{k-1} w_{\mathcal{A}}(i, i+1) = \sum_{i=1 }^{k-1} \left( u_i(A_{i+1}) - u_i(A_i) \right)
&= \sum_{i=1 }^{k-1} \sum_{t = 1}^T \left( u_i(\mu^t_{i+1}) - u_i(\mu^t_i) \right) \\
&= \sum_{i=1 }^{k-1} \sum_{t = 1}^T w_{\mu^t}(i, i+1) \\
&= \sum_{t = 1}^T w_{\mu^t} (P)
\end{align*}

Observe that for any  round $1\leq t \leq T-1$, the quantity $w_{\mu^t} (P)$ can be bounded from above by the difference in utility of agent $k$'s favorite available item between rounds $t$ and $t+1$. To see this, consider an alternative matching $\widetilde{\mu^t}$ that is feasible in round $t$, defined as follows:
$$
\widetilde{\mu^t_i} := \begin{cases}
 \argmax\limits_{g\in J^{t+1} } u_k (g) \quad & \text{ if } i=k  \\
 \mu^t_{i+1} \quad & \text{ if } i\in \{1,...,k-1\} \\
 \mu^t_{i} \quad & \text{ if } i\not\in \{1,...,k\}
\end{cases}
$$
Note that the feasibility of $\widetilde{\mu^t}$ follows since every matched item is available in round $t$. Since $\mu^t$ is max weight matching, we have
\begin{align*}
\sum_{i\in N} u_i(\widetilde{\mu^t_i} ) - \sum_{i\in N} u_i(\mu^t_i )  &= \max\limits_{g\in J^{t+1}} u_k (g) - u_k (\mu^t_k) + \sum_{i=1}^{k-1} (u_i(\mu^t_{i+1})-u_i(\mu^t_{i})) \\
& = \max\limits_{g\in J^{t+1} } u_k (g) - u_k (\mu^t_k) + w_{\mu^t} (P) \leq 0
\end{align*}
Moreover, since $\mu^t_k \in J^{t}$, we have $u_k(\mu^t_k) \leq \max_{g\in J^{t}} u_k(g)$. Combining these inequalities gives
\[
w_{\mu^t} (P) \leq \max_{g\in J^{t}} u_k (g) - \max\limits_{g\in J^{t+1} } u_k (g).
\]

For the final round $T$, we use a similar argument, but modifying the alternate feasible matching by setting $\widetilde{\mu^T_k}=\mu^T_1$ and $\widetilde{\mu^T_i}=\mu^T_{i+1}$ for each $i\in \{1,...,k-1\}$, while keeping all other assignments unchanged. This gives,
\[
w_{\mu^T} (P) \leq  u_k (\mu^T_k) -  u_k (\mu^T_1).
\]
Combining the above inequalities, we obtain
\begin{align*}
    w_{\mathcal{A}}(P) = \sum_{t=1}^T w_{\mu^t}(P) &\leq u_k (\mu^T_k) -  u_k (\mu^T_1) + \sum_{t=1}^{T-1}  \max_{g\in J^{t} } u_k (g) - \max\limits_{g\in J^{t+1} } u_k (g)  \\
    &= u_k (\mu^T_k)  -  u_k (\mu^T_1)+ \max_{g\in J^{1} } u_k (g) - \max_{g\in J^{T} } u_k (g)  \\
    &= -(\max_{g\in J^{T} } u_k (g) - u_k (\mu^T_k)  ) + (\max_{g\in J^{1} } u_k (g) -u_k (\mu^T_1)) \\
    & \leq -u_k (\mu^T_1) \leq 1
\end{align*}
where the last inequality follows since $\max_{g\in J^{T} } u_k (g) \geq u_k (\mu^T_k)$, and since each item is an objective chore  $\max_{g\in J^{1} } u_k (g) <0$.

Finally, note that since the allocation in each round is determined by a maximum-weight matching, it is envy-freeable. Moreover, the sum of envy-freeable allocations is itself envy-freeable, and therefore the allocation $\mathcal{A}$ is envy-freeable. By the result of \citet{HalpernSh19}, setting $p_i\ = \ell_\alloc(i)$  for each $i\in N$ yields an  envy-free allocation.  Since the weight of any path in the envy graph is bounded above by one, we have the desired subsidy bound.

\subsection{Proof of Lemma~\ref{lem:algoOneProperties}}
\label{App:Algone}

We first check that each element~$M_j \in G$ is indeed a meta-good by showing that its interest set~$T_j$ is non-empty.
The non-emptiness of~$T_j$ is guaranteed because Phase~2 only merges the two subjective goods $g_1$ and $g_2$ with a chore~$c$ when some agent values the whole bundle $g_1 \cup g_2 \cup \{c\}$ non-negatively (line~\ref{alg:IfMerge}).
If no chore is added to the newly formed bundle (line~\ref{alg:ElseMerge}), then the selected agent~$i$, who values both $g_1$ and $g_2$ non-negatively, also values their union non-negatively.
Thus, in each iteration of the \verb|while|-loop in Phase~2, the newly formed subjective good remains a subjective good.
As a result, when the algorithm terminates, for every element of~$G$, there is at least one agent who values it non-negatively.

\medskip
\noindent\textbf{Property~\ref{prop:POne}: Disjoint interest sets.}
In Phase~2, we merge whenever some agent has degree at least~$2$ in $R[N, U]$.
Each merging operation replaces two neighbors of that agent by a single object, so $|U|$ strictly decreases and the loop terminates.
If $Z_{\mathrm{rem}} \neq \varnothing$, the termination condition of the \verb|while|-loop in Phase~2 implies that $\deg(i) \le 1$ for all agents~$i$ in the final interest graph $R[N, G]$.
Equivalently, each agent belongs to at most one interest set, that is, $T_i \cap T_j = \varnothing$ for all~$i \neq j$.

(When $Z_{\mathrm{rem}} = \varnothing$, disjointness may not be guaranteed.)

\medskip
\noindent\textbf{Properties~\ref{prop:PTwo} and~\ref{prop:PThree}: Upper bound on interested values and chore-maximality.}

If $Z \neq \emptyset$ when Phase~2 begins, both \ref{prop:PTwo} and \ref{prop:PThree} hold.
This is because, first, no subjective goods have yet been merged, implying~\ref{prop:PTwo}.
Moreover, if \ref{prop:PThree} did not hold, then we would have a subjective good which is not chore-maximal, meaning that the \verb|while|-loop of Phase~1 would still be running.

In the following, we show that any execution of Phase~2 preserves both \ref{prop:PTwo} and~\ref{prop:PThree}.
Phase~2 only executes if $Z \neq \emptyset$, meaning that at least one objective chore is available.
We now check that the invariants are preserved under the two possible bundling operations (lines~\ref{alg:IfMerge} and \ref{alg:ElseMerge}).
\begin{itemize}
\item In the \verb|if| case (line~\ref{alg:IfMerge}), the item $g_{\mathrm{new}}$ must be chore-maximal since for all agent~$j \in N$ and objective chore~$d \in Z$,
\[
u_j(g_{\mathrm{new}} \cup\{d\}) = u_j(g_1) + u_j(g_2) + u_j(c) + u_j(d) = \bigl( u_j(g_1) + u_j(c) \bigr) + \bigl( u_j(g_2) + u_j(d) \bigr) < 0.
\]

Likewise for the upper bound: $\forall j \in N$,
\[
u_j(g_{\mathrm{new}}) = u_j(g_1) + u_j(g_2) + u_j(c) = u_j(g_1) + \bigl( u_j(g_2) + u_j(c) \bigr) < u_j(g_1) \leq 1.
\]

\item In the \verb|else| case (line~\ref{alg:ElseMerge}), the branch condition gives, for every agent~$j$ and every remaining chore~$c$,
\[
u_j(g_{\mathrm{new}})+u_j(c)<0.
\]
Thus $g_{\mathrm{new}}$ is chore-maximal. Moreover, since $u_j(c)\ge -1$, every interested agent~$j$ satisfies
\[
u_j(g_{\mathrm{new}})<-u_j(c)\le 1,
\]
which proves the required upper bound.
\end{itemize}

\end{document}